\documentclass{article}
\usepackage{amssymb,amsfonts,amsthm}
\usepackage{amsmath}
\usepackage[mathscr]{eucal}
\usepackage{wasysym}

\def\pmb#1{\setbox0=\hbox{$#1$}%
  \kern-.025em\copy0\kern-\wd0
  \kern.05em\copy0\kern-\wd0
  \kern-.025em\raise.0433em\box0}
\def\pmbs#1{\setbox0=\hbox{$\scriptstyle #1$}%
  \kern-.0175em\copy0\kern-\wd0
  \kern.035em\copy0\kern-\wd0
  \kern-.0175em\raise.0303em\box0}
\def\be{\begin{equation}}
\def\ee{\end{equation}}
\def\bea{\begin{eqnarray}}
\def\eea{\end{eqnarray}}

\def\vec#1{\mbox{\boldmath$#1$}}

\def\Om{\Omega}

\def\Udot{\dot{U}}

\def\bOm{\mbox{\boldmath $\Omega$}}
\def\bna{\mbox{\boldmath $\nabla$}}

\def\vece{\vec{e}}

\def\ptl{\partial}
\def\parb{\pmb{\partial}}
\def\la{\langle}
\def\ra{\rangle}
\def\hsp5{\hspace{5mm}}
\def\case#1/#2{\textstyle\frac{#1}{#2}}

\theoremstyle{plain}
\newtheorem{theorem}{Theorem}[section]
\newtheorem{corollary}[theorem]{Corollary}
\newtheorem{lemma}[theorem]{Lemma}
\newtheorem{proposition}{Proposition}[section]
\newtheorem*{conjecture}{Conjecture}

\theoremstyle{remark}

\newcommand{\textfrac}[2]{{\textstyle \frac{#1}{#2}}}

\newcommand{\enl}{\\\hfill\rule{0pt}{0pt}}

\title{\sc Perfect fluids and generic spacelike singularities}
\author{\sc
 Patrik Sandin $^{1}$\thanks{Electronic address: {\tt patrik.sandin@kau.se}}\
,\ \ and Claes Uggla$^{1}$\thanks{Electronic address:
{\tt claes.uggla@kau.se}}\\
$^{1}${\small\em Department of Physics, University of Karlstad,}\\
{\small\em S-651 88 Karlstad, Sweden}}

\begin{document}
\maketitle

\begin{abstract}

We present the 1+3 Hubble-normalized conformal orthonormal
frame approach to Einstein field equations, and specialize it
to a source that consists of perfect fluids with general
barotropic equations of state. We use this framework to give
specific mathematical content to conjectures about generic
spacelike singularities that were originally introduced by
Belinskii, Khalatnikov, and Lifshitz. Assuming that the
conjectures hold, we derive results about how the properties of
fluids and generic spacelike singularities affect each other.

\end{abstract}

\centerline{\bigskip\noindent PACS number(s): 04.20.-q, 98.80.Hw,
98.80.Dr, 04.20.Jb} \vfill
\newpage

\section{Introduction}\label{intro}

Although the singularity theorems say little about the nature
of singularities, the very definition of a singularity implies
that there exists a variable scale---the affine parameter
distance from/to the singularity of a causal inextendible
geodesic that is used to define it, furthermore, the one
dynamical input that goes into the theorems, the Raychaudhuri
equation for the expansion $\theta$,\footnote{In the case of
timelike geodesics; in the null geodesic case an analogous
equation plays a similar role.} also implies a variable scale
given by the expansion itself, since $\theta$ has unit
(time)$^{-1}$ (or, equivalently, (length)$^{-1}$, since we set
the speed of light $c$ to one). In this paper we study detailed
asymptotic dynamical aspects of generic singularities, and this
brings the expansion and the coupling of the Raychaudhuri
equation to the remaining Einstein's equations into focus. We
will locate the singularity in the past and we therefore refer
to it as a `cosmological' singularity. Since we study
asymptotic temporal developments, we consider timelike
reference congruences for which $\theta>0$ in the vicinity of
the singularity, where $\theta \rightarrow + \infty$
asymptotically, i.e., we are interested in `crushing'
singularities. Furthermore, due to the `cosmological' context
we will replace $\theta$ with the Hubble variable $H$ which is
defined as $H=\frac{1}{3}\,\theta$ (note that it is common in
FRW cosmology to refer to $H^{-1}$ as a characteristic time
scale, also known as the Hubble radius when referred to as a
length scale).

The asymptotic blow up of $H$ suggests that we should
asymptotically `factor out' $H$, and thereby the associated
variable scale, toward the singularity, preferably so that the
two following desirable features are incorporated into the
formalism:
\begin{itemize}
\item[(i)] Preservation of causal structure, since it is
    reasonable to believe that there is a close connection
    between causal structure and the nature of
    singularities.
\item[(ii)] Adaption to scale-invariance, since there are
    many known as well as conjectured links between
    scale-invariant, i.e., self-similar, solutions and
    asymptotic properties of many types of singularities.
\end{itemize}
The natural way to accomplish this is by means of a conformal
transformation (satisfies (i)) with a conformal factor that
involves $H$ (factoring out of $H$) so that the key variables
are (conformally) scale-invariant, i.e., dimensionless, and
thus adapted to the properties of self-similar solutions, since
such solutions are scale-invariant (satisfies (ii)). Hence we
use a {\em conformally Hubble-normalized scale-invariant
formulation\/} based on the conformal transformation
\be {\bf G} = H^{2}{\bf g} \qquad \Leftrightarrow \qquad {\bf
g} = H^{-2}{\bf G}, \ee
where we assume that $H>0$ in the vicinity of the singularity;
${\bf g}$ is the physical metric, which like $H^{-2}$ naturally
carries dimension (length)$^2$, and hence it follows that the
unphysical metric ${\bf G}$ is dimensionless. Because of this,
scalars constructed from ${\bf G}$ take constant finite values
for self-similar models that admit spacetime transitive
homothetic symmetry groups. This leads to a major advantage:
Asymptotically bounded variables for a system of coupled
regularized field equations.

We also find it advantageous to express the field equations as
a system of first order partial differential equations. A
natural way to do this within the conformally Hubble-normalized
scale-invariant context is to use the Hubble-normalized
Conformal OrthoNormal Frame approach (subsequently shortened to
the acronym CONF). In this approach one chooses a frame field
that is orthonormal to the dimensionless metric ${\bf G}$, and
not to the physical metric ${\bf g}$, i.e., we introduce
Hubble-normalized conformal orthonormal vector fields $\parb_a$
that are dual to $\bOm^a$, i.e.,
$\langle\,\bOm^a,\,\parb_b\,\rangle = \delta^{a}{}_{b}$, such
that
\begin{equation} \label{defconfon}
{\bf g} = H^{-2}\,{\bf G} = H^{-2}\,\eta_{ab}\,\bOm^a\,\bOm^b\:,
\end{equation}
where $\eta_{ab}={\rm diag}(-1,1,1,1)$, and $a,b=0,1,2,3$.
Furthermore, since we are interested in asymptotic temporal
behavior, we let $\parb_0$ be tangential to the reference
congruence, i.e., $\parb_0 \propto \partial/\partial x^0$,
where $x^0$ is the time coordinate along the reference
congruence (see Appendix~\ref{app}). This naturally leads to
the 1+3 Hubble-normalized CONF formulation, which is a
specialization of the 1+3 CONF formulation, introduced
in~\cite{rohugg05}, to a conformal factor related to $H$ as
described above; we present the 1+3 Hubble-normalized CONF
field equations in Appendix~\ref{app}.

In this paper we consider a source that consists of several
perfect fluids. The $i$:th perfect fluid yields a stress-energy
tensor component,
\be T^{ab}_{(i)} = (\tilde{\rho}_{(i)} + \tilde{p}_{(i)})
\tilde{u}^a_{(i)}\tilde{u}^b_{(i)} + \tilde{p}_{(i)} g^{ab},
\ee
to the total stress-energy tensor, $T^{ab}=\sum_i
T^{ab}_{(i)}$, where $\tilde{\rho}_{(i)}$ and $\tilde{p}_{(i)}$
are the energy density and pressure, respectively, in the rest
frame of the $i$:th fluid, while $\tilde{u}^a_{(i)}$ is its
4-velocity; throughout we assume that $\tilde{\rho}_{(i)}\geq
0$. It is natural to make a 1+3 split of $\tilde{u}^a_{(i)}$
w.r.t. the vector field $u^a$ that is tangential to the
reference congruence, and introduce a 3-velocity $v^a_{(i)}$
according to\footnote{One reason for why this is convenient is
that since the components of the 3-velocity $v^\alpha_{(i)}$ in
the orthonormal frame of ${\bf g}$ are dimensionless they
coincide with the 3-velocity components of the conformal
4-velocity in the Hubble-normalized frame of ${\bf G}$.}
\bea \tilde{u}^a_{(i)} = \Gamma_{(i)}(u^a + v^a_{(i)}); \qquad
u_a v^a_{(i)}=0,\qquad \Gamma_{(i)} = 1/\sqrt{1 - v^2_{(i)}}.
\eea
The $i$:th fluid is, apart from its 3-velocity, conveniently
characterized by its energy-density w.r.t. $u^a$, $\rho_{(i)}$,
which is defined in terms of $\tilde{\rho}_{(i)}$ and
$v_{(i)}^2$ according to
\be \rho_{(i)} =
\Gamma^{2}_{(i)}\,G_+^{(i)}\,\tilde{\rho}_{(i)} ,\qquad
G_\pm^{(i)} = 1 \pm w_{(i)}\, v_{(i)}^2 ,\qquad w_{(i)} =
\frac{\tilde{p}_{(i)}}{\tilde{\rho}_{(i)}} . \ee
Throughout, we are going to assume that the perfect fluids
satisfy barotropic equations of state, i.e., $\tilde{p}_{(i)} =
\tilde{p}_{(i)}(\tilde{\rho}_{(i)})$, and hence
$w_{(i)}({\tilde{\rho}}_{(i)})$; special cases of interest are
dust, $w=0$, radiation, $w=\frac{1}{3}$, and stiff fluids,
$w=1$.

To conform with standard convention in cosmology, we
Hubble-normalize $\rho_{(i)}$ as follows
\be \Omega_{(i)} = \frac{\rho_{(i)}}{3H}^2 ,\ee
where $H$ is the Hubble variable associated with $u^a$. A 1+3
irreducible Hubble-normalized decomposition of the
stress-energy tensor, see Appendix~\ref{app}, yields
\be\label{pfrel} Q_{(i)}^\alpha = (1 + w_{(i)})
(G^{(i)}_+)^{-1}\, \Omega_{(i)}\, v_{(i)}^\alpha ;\,\, P_{(i)}
= w_{(i)}\Omega_{(i)} + \textfrac{1}{3} (1 -
3w_{(i)})Q^{(i)}_\alpha v_{(i)}^\alpha;\,\,
\Pi^{(i)}_{\alpha\beta} =
Q^{(i)}_{\la\alpha}v^{(i)}_{\beta\ra}, \ee
where $Q_{(i)}^\alpha, P_{(i)}, \Pi^{(i)}_{\alpha\beta}$ are
the Hubble-normalized components of the energy flux, pressure,
and stress tensor, respectively (w.r.t the temporal reference
congruence), and hence the Hubble-normalized stress-energy
tensor of the $i$:th fluid is characterized by $\Omega_{(i)}$,
$v^\alpha_{(i)}$, and $w_{(i)}$ ($\alpha,\beta=1,2,3$).

The 1+3 Hubble-normalized CONF formulation of the field
equations for $I$ fluids involve the following quantities, see
Appendix~\ref{app},
\begin{itemize}
\item[] Hubble-normalized frame variables: $\{{\cal M},\,
    {\cal M}_\alpha,\, E_\alpha{}^i\}$.
\item[] Hubble-normalized connection/commutator variables:
    $\{W^\alpha,\, \Udot^\alpha,\, R^\alpha,\,
    \Sigma_{\alpha \beta},\, A_\alpha,\,
    N_{\alpha\beta}\}$.
\item[] Hubble-normalized perfect fluid variables:
    $\{\Omega_{(1)},\, v_{(1)}^\alpha,....,\,
    \Omega_{(i)},\,v_{(i)}^\alpha,....,\,
    \Omega_{(I)},\,v_{(I)}^\alpha\}$.
\end{itemize}
Apart from the already described fluid quantities, the
quantities ${\cal M}$ and ${\cal M}_\alpha$ are the
Hubble-normalized threading lapse and shift functions
respectively, while $E_\alpha{}^i$ are the Hubble-normalized
spatial frame components. The quantities
$W^\alpha,\,\Udot^\alpha,\, \Sigma_{\alpha \beta}$ describe the
vorticity, acceleration, and shear of the Hubble conformal
reference congruence, while $R^\alpha$ describes the rotation
of the spatial frame w.r.t. a Fermi frame in the space
orthogonal to the reference congruence. Finally $A_\alpha$ and
$N_{\alpha\beta}$ gives the commutator functions (or,
equivalently, the spatial Hubble-conformal connection
coefficients) of the Hubble-normalized spatial frame.

The outline of the paper is as follows. In the next section we
formulate two conjectures, originally introduced by Belinskii,
Khalatnikov, and Lifshitz, in terms of the above
Hubble-normalized variables. In Section~\ref{paststab} we
explore the consequences of these conjectures and derive a
number of results concerning the past stability and instability
on the so-called silent boundary; in particular we present the
past attractors on the silent boundary (note that since there
exists a dynamical one-to-one correspondence between the silent
boundary and the spatially homogeneous models, it follows that
our results also pertain to the latter case). Then the
stability of these results are investigated in the context of
the full state space in Section~\ref{stabfull}, where we also
discuss possible temporal gauge choices, and in particular if
it is possible to use fluid congruences as temporal reference
congruences to describe so-called asymptotically silent and
local singularities; we find that this is only possible if
there exist fluids with a sound speed that is equal to or
larger than the speed of light. We conclude with a summary and
some remarks about our results in Section~\ref{concl}, together
with some comments about some open issues. Appendix \ref{app}
establishes conventions and notation by giving the 1+3
Hubble-normalized CONF field equations explicitly. Finally we
describe a number of important subsets in
Appendix~\ref{invbound}.

\section{BKL conjectures}\label{conjecture}

In~\cite{BKL82} p656 Belinskii, Khalatnikov, and Lifshitz made
the following important conjecture:
\begin{conjecture}\label{BKLL*}
``...in the asymptotic vicinity of the singular point the
Einstein equations are effectively reduced to a system of
ordinary differential equations with respect to time: the
spatial derivatives enter these equations `passively' without
influencing the character of the solution.''
\end{conjecture}

In the present 1+3 Hubble-normalized CONF framework we
reformulate this conjecture in terms of two intertwined
conditions:
\begin{conjecture} [1]
\label{BKLL}
\begin{subequations}
\begin{eqnarray}
&\textit{(a) Asymptotic surface formation\/}: &
\lim_{x^0 \to -\infty}({\cal M}_\alpha,\, W_\alpha,\, \Udot_\alpha,\,
r_\alpha) = 0, \nonumber \\
& & \quad 0 < C_1 \leq \lim_{x^0 \to -\infty}{\cal M} \leq C_2 < \infty . \label{condasila}\\
&\textit{(b) Asymptotic locality condition\/}: & \lim_{x^0 \to
-\infty} (E_{\alpha}{}^{i},\, \parb_\alpha\,{\bf X})
= 0, \label{condasilb}
\end{eqnarray}
\end{subequations}
\end{conjecture}
\noindent where $C_{1,2}=\mathrm{const}$ and ${\bf X} = ({\cal
M},\, {\cal M}_\alpha,\, W_\alpha,\,
\Udot^\alpha,\,R^\alpha,\,\Sigma_{\alpha \beta},\, A_\alpha,\,
N_{\alpha\beta},\,\Omega_{(1)},\,
v_{(1)}^\alpha,\,...,\,\Omega_{(I)},\, v_{(I)}^\alpha)$, where
$I$ denotes the number of perfect fluids.

Condition (a) implies that the spatial frame is asymptotically
hypersurface forming; in addition $\Udot_\alpha=0$ implies that
the timelike congruence is conformally geodesic, which amounts
to an inverse mean curvature flow for the original physical
spacetime, while $r_\alpha=0$ implies that a foliation is a
constant mean curvature foliation in the physical spacetime.
Furthermore, $\Udot_\alpha=r_\alpha=0$ implies that the
reference timelines are geodesics in the original physical
spacetime. Finally, the condition that ${\cal M}$ is
asymptotically bounded implies that the reference congruence
asymptotically gives rise to a foliation that yields a
simultaneous bang function when $x^0\to -\infty$.\footnote{This
follows from that we can reparameterize the congruence
according to ${\bar x}^0 = \exp(x^0)$, so that the past
singularity occurs at ${\bar x}^0=0$. However, we believe that
the condition on ${\cal M}$ can be weakened, which, however, we
have refrained from doing in order to keep the discussion
reasonably simple. The central restriction is to choose a
reference congruence that asymptotically yields a spacelike
foliation that has a simultaneous bang function, cf. the
discussion about the synchronous gauge in~\cite{BKL82}.}

The field equations for the 1+3 Hubble-normalized
CONF-variables for a source that consists of several perfect
fluids are given by equations~\eqref{devoleq},
\eqref{dconstreq}, and~\eqref{perf} in Appendix~\ref{app}. This
system admits an invariant subspace called the \emph{silent
boundary\/}, see Appendix~\ref{invbound}, which is
characterized by
\begin{equation}
({\cal M}_\alpha,\,W_\alpha,\, \Udot_\alpha,\,
r_\alpha,\, E_{\alpha}{}^{i}) = 0 .
\end{equation}
As discussed in Appendix~\ref{app} the equations on this subset
form a coupled system of ODE that is identical to the system
that describes the dynamics of spatially homogeneous Bianchi
models. If conjecture (1) holds along a timeline we may say
that the dynamics become \emph{asymptotically local\/}, since
the dynamics then is asymptotically described by the silent
boundary on which the dynamics for a timeline is governed by
what happens along the timeline alone, and hence it would
perhaps be more appropriate to refer to the silent boundary as
the \emph{local boundary\/}.\footnote{Asymptotic silence is
defined as the formation of particle horizons that shrink to
zero size in all directions along any timeline that is not
asymptotically null toward the singularity, thus asymptotically
prohibiting communication. When the nomenclature silent
boundary was introduced it was believed that asymptotic silence
implied asymptotical local dynamics, however, the discovery of
recurring spike formation~\cite{andetal05,limetal09} shows that
this is not the case. Although we expect that conjecture (1)
holds for most timelines for an open set of solutions, the
existence of recurring spike formation suggests that there
exist generic singularities with special timelines with
`non-BKL' behavior for which conjecture (1) does not hold.}
Finally, a singularity that obeys `the locality conjecture' (1)
will be referred to as an \textit{asymptotically local
singularity\/}.

The dynamical relevance of the silent/local boundary depends on
if the conditions in conjecture (1) holds. From now on we will
assume that this is the case and derive the consequences of
this assumption. A \emph{necessary\/} condition for the
dynamics of a timeline to approach the silent/local boundary is
that $E_\alpha{}^i\rightarrow 0$ toward the past singularity,
which is equivalent to that the conformally Hubble-normalized
contravariant spatial 3-metric
    \begin{equation}
    {}^3\!G^{ij} = \delta^{\alpha\beta}\,E_\alpha{}^i\,E_\beta{}^j
    \end{equation}
tends to zero. Due to~\eqref{deeq}, ${}^3\!G^{ij}$ satisfies
the equation
\be\label{Gij}
\parb_0 {}^3\!G^{ij} = 2(q\delta^{\alpha\beta} -
\Sigma^{\alpha\beta})E_\alpha{}^i\,E_\beta{}^j . \ee

The vanishing of $E_\alpha{}^i$ to the past is equivalent to
the condition that the time integral to the initial singularity
of the eigenvalues of the matrix $(q\delta_\alpha{}^\beta -
\Sigma_\alpha{}^\beta)$ negatively diverge for all of the
eigenvalues,
\begin{equation}\label{positivity condition}
\int_{\tilde{x}^0}^{-\infty}{\rm eig}(q\delta_\alpha{}^\beta -
\Sigma_\alpha{}^\beta){\cal M}dx^0 = -\infty ,
\end{equation}
where $\tilde{x}^0$ describes some reference point along the
timeline. It has implicitly been shown that the
condition~\eqref{positivity condition} is fulfilled for vacuum
and orthogonal fluid Bianchi models of type
$IX$~\cite{rin01,heiugg09a,heiugg09b}, but this, of course,
does not imply that it is true in general. However, we will in
the following assume that condition~\eqref{positivity
condition} holds and work out the consequences of that
assumption, which leads to a consistent picture. Our assumption
that there exist generic dynamics that is asymptotically
described by the silent boundary suggests an analysis in two
steps:
\begin{itemize}
\item[1.] Identification of the past attractor\footnote{The
    past attractor of a dynamical system given on a state
    space $X$ is defined as the smallest closed invariant
    set $\mathcal{A}^- \subseteq \overline{X}$ such that
    the $\alpha$-limits of all $p \in X$, apart from a set
    of measure zero, satisfy $\alpha(p) \subseteq
    \mathcal{A}^-$~\cite{mil85}.} on the silent boundary.
\item[2.] Perturbation of the past attractor in the full
    inhomogeneous state space to establish if it is stable
    or not.
\end{itemize}

A proof that identifies the attractor and shows its stability
in the full infinite dimensional state space amounts to a proof
of a singularity theorem that concerns the details of a generic
singularity. This is likely to be an extremely hard problem,
and we will therefore only provide proofs about some aspects in
the context of that our `BKL-like' assumptions hold.

There exists a second conjecture proposed by Belinskii,
Khalatnikov, and Lifshitz~\cite{LK63,BKL70} that is relevant in
this context: the asymptotic `matter does not matter'
conjecture.
\begin{conjecture}\label{BKLM*}
For a typical cosmological model, the matter content is not
dynamically significant near the initial singularity.
\end{conjecture}

In our case of a source of several perfect fluids we formulate
this conjecture in terms of our variables as:
\begin{conjecture}[2]\label{BKLM}
\begin{equation}
\lim_{x^0 \to -\infty} \Omega_\mathrm{tot} = 0,
\end{equation}
\end{conjecture}
\noindent where the total Hubble-normalized energy-density of
the source is given by
\be \Omega_\mathrm{tot} = \sum_i \Omega_{(i)} . \ee
Note that $\Omega_\mathrm{tot}=0$ implies
$\Omega_{(i)}=0\,\forall i$, since $\tilde{\rho}_{(i)}\geq 0$,
and from this it follows that the entire total
Hubble-normalized stress-energy tensor is zero. Hence `the
matter does not matter' conjecture (2) asserts that the
Hubble-normalized stress-energy tensor asymptotically
approaches zero toward the singularity. As pointed out by BKL
themselves in the context of one fluid, this is not to be
expected for all equations of state, e.g., not for a stiff
fluid.

\section{Past stability and instability on the silent boundary}\label{paststab}

To proceed with step 1 we first give the equations on the
silent boundary, which are obtained by restricting the full
system of equations, \eqref{devoleq}, \eqref{dconstreq},
and~\eqref{perf} in Appendix~\ref{app}, to the silent boundary
invariant subspace, $({\cal M}_\alpha,\,W_\alpha,\,
\Udot_\alpha,\, r_\alpha,\, E_{\alpha}{}^{i}) = 0$. As
discussed in Appendix~\ref{invbound}, the equations on the
silent boundary are the same as in the spatially homogeneous
case, and hence all results in this section also pertain to
these models.

\subsection{Equations on the silent boundary}\label{dynsil}

Instead of the peculiar 3-velocity $v^\alpha$ we find it useful
to introduce $v\geq 0$ and the unit vector
$c^\alpha=v^\alpha/v$ as variables. This leads to the following
state vector on the silent boundary:
\begin{equation}\label{statevec}
{\bf S} = (\Sigma_{\alpha\beta},\,A_\alpha,\,N^{\alpha\beta}) \oplus
(\Omega_{(1)},\,v_{(1)},\,c^\alpha_{(1)}) \oplus ... \oplus
(\Omega_{(I)},\,v_{(I)},\,c^\alpha_{(I)}) .
\end{equation}
Note that we have not included $R_\alpha$ in the state vector
since there exists no evolution equation for $R_\alpha$, which
is due to that $R_\alpha$ represents the freedom to rotate the
spatial frame (nor ${\cal M}$, which represents the freedom to
reparameterize the reference timelines). On the silent boundary
we have the following evolution equations and constraints that
govern the dynamics of ${\bf S}$.

\vspace*{2mm} \noindent {\em Evolution equations}:
\begin{subequations}\label{devoleqsil}
\begin{align}
\parb_0\, \Sigma_{\alpha\beta} &= -(2-q)\Sigma_{\alpha\beta} +
2\epsilon^{\gamma\delta}{}_{\la \alpha}\,\Sigma_{\beta\ra
\delta}\,R_\gamma - \,{}^3\!{\cal S}_{\alpha\beta} +
3\Pi_{\alpha\beta},\label{sigsil}\\
\parb_0\,A_\alpha &= F_\alpha{}^\beta\, A_\beta,\label{Asil}\\
\parb_0\, N^{\alpha\beta} &=
(3q\delta_\gamma{}^{(\alpha} - 2F_\gamma{}^{(\alpha}) N^{\beta )\gamma},\label{Nsil}\\
\parb_0\,\Omega &= (2q - 1 - 3w)\,\Omega + [(3w-1)\,v_\alpha -
\Sigma_{\alpha\beta}\,v^\beta + 2A_\alpha]\,Q^\alpha,\label{ompfsil}\\
\parb_0 v &= \bar{G}_-^{-1}\,(1-v^2)\,
\left[3c_s^2 - 1 -2\,c_s^2\,A^\beta\,c_\beta\,v -
\Sigma_{\alpha\beta}\,c^\alpha c^\beta \right] v,\label{pec}\\
\parb_0 c_{\alpha} &= - [\delta_\alpha{}^\beta - c_\alpha c^\beta]
[\Sigma_\beta{}^\gamma\,c_{\gamma} + v\,A_\beta +
\epsilon_{\beta}{}^{\gamma\delta}\,(R_{\delta} +
v\,N_{\delta}{}^{\nu}\,c_{\nu})\,c_{\gamma}].  \label{cpec}
\end{align}
\end{subequations}

\vspace*{2mm} \noindent {\em Constraint equations}:
\begin{subequations}\label{dconstreqsil}
\begin{align}
0 &= 1 - \Sigma^2 - \Omega_k - \Omega,\label{dGausssil}\\
0 &= (3\delta_\alpha{}^\gamma\,A_\beta +
\epsilon_{\alpha\delta}{}^{\gamma}
\,N^\delta{}_\beta)\,\Sigma^\beta{}_\gamma - 3Q_\alpha,\label{dCodazzisil}\\
0 &= A_\beta\, N^\beta{}_\alpha,
\end{align}
\end{subequations}
where
\begin{subequations}\label{threecurvsil}
\begin{alignat}{2}
F_{\alpha}{}^{\beta} &=  q\,\delta_{\alpha}{}^{\beta} -
\Sigma_{\alpha}{}^{\beta} - \epsilon_{\alpha}{}^{\beta}{}
_{\gamma}\,(W^{\gamma}+R^{\gamma}), & \qquad
q &= 2\Sigma^{2} + \textfrac{1}{2}(\Om+3P) ,\label{qsil}\\
{}^3\!{\cal S}_{\alpha\beta} &= B_{\la \alpha\beta \ra} +
2\epsilon^{\gamma\delta}{}_{\la
\alpha}\,N_{\beta\ra\delta}\,A_\gamma, &\qquad B_{\alpha\beta}
&= 2 N_{\alpha\gamma}\,N^\gamma{}_\beta -
N^\gamma{}_\gamma\,N_{\alpha\beta},\\
{}^3\!{\cal R} &= -\textfrac{1}{2}B^\alpha{}_\alpha - 6A^2,
&\qquad \Omega_k &= -\textfrac{1}{6}\,{}^3\!{\cal R},\\
\bar{G}_- &= 1 - c_s^2\,v^2, & \qquad c_s^2 &=
\frac{d\tilde{p}}{d\tilde{\rho}} ,
\end{alignat}
\end{subequations}
where $\Sigma^2 =
\frac{1}{6}\Sigma_{\alpha\beta}\Sigma^{\alpha\beta}$, and where
$c_s^2$ can be interpreted as the speed of sound when
non-negative. Note that it is the complete
stress-energy-momentum objects that appear in~\eqref{sigsil},
\eqref{dGausssil}, \eqref{dCodazzisil}, and~\eqref{qsil}, while
the perfect fluid equations~\eqref{ompfsil} and~\eqref{pec},
\eqref{cpec} describe the dynamics of an individual perfect
fluid component, where we have dropped the index $(i)$ to avoid
cluttered notation. To obtain the perfect fluid equations we
have assumed that the Hubble-normalized interactions between
the different fluids are asymptotically zero, see
Appendix~\ref{app}. It follows from~\eqref{pec} that $v=0$ is
an invariant subset and so is $v=1$ when $c_s^2\neq 1$, i.e.,
when the equation of state of the fluid is not stiff.
Remarkably, $w$ does not appear in the peculiar velocity
equations~\eqref{pec} and \eqref{cpec}, nor do $\Omega$ and $q$
---the equation of state enters via $c_s^2$ only, and thus a
general barotropic equation of state leads to formally the same
expressions as that of a linear equation of state! However, in
general $c_s^2$ is a function of a suitable matter variable,
e.g. $c_s^2(\tilde{\rho})$, while $c_s^2=w=const$ in the linear
case. Moreover, the equation for the peculiar velocity
direction $c_\alpha$, i.e.~\eqref{cpec}, contains neither
$c_s^2$ nor $w$, i.e., it contains no \textit{direct} coupling
to the equation of state at all!

For completeness we here give the evolution equation for the
peculiar velocity $v_\alpha$ on the silent boundary:
\begin{align}
\parb_0 v_{\alpha} &= \bar{G}_-^{-1}\,\left[ (1-v^2)(3c_s^2 - 1 -
c_s^2\,A^\beta\,v_\beta) + (1-c_s^2)(A^\beta +
\Sigma_\gamma{}^\beta\,v^\gamma)\,v_\beta \right] v_\alpha
\nonumber \\
& \quad - [\Sigma_\alpha{}^\beta +
\epsilon_\alpha{}^{\beta\gamma}\,(R_\gamma +
N_\gamma{}^\delta\,v_\delta)]\,v_\beta - A_\alpha\,v^2 .\label{pecsil}
\end{align}
It is of interest to also give the evolution equations for
$\rho$ and $\tilde{\rho}$, for a fluid component, on the silent
boundary (i.e., let $(E_{\alpha}{}^{i}, {\cal M}_\alpha,
W_\alpha, \Udot_\alpha, r_\alpha) = 0$ in the equations for
these objects):
\begin{subequations}\label{energydensevol}
\begin{align}
\parb_0\, ({\rm ln}\,\rho) &= -(1+w) G_+^{-1}[3  - 2A_{\alpha}
v^\alpha + (v^2 + \Sigma_{\alpha\beta}\, v^\alpha v^\beta)],
\label{rhosil}\\
\parb_0 \, ({\rm ln}\, \tilde{\rho}) &= -(1+w)\bar{G}_-^{-1}
[3 - 2A_{\alpha}v^{\alpha} - (v^2 + \Sigma_{\alpha\beta}\, v^{\alpha}
v^{\beta})],\label{tilderhosil}
\end{align}
\end{subequations}
where we again have dropped the index $(i)$ in~\eqref{pecsil}
and~\eqref{energydensevol} to avoid cluttered notation.

\subsection{Past evolution on the silent boundary}\label{paststabsil}

\begin{proposition}\label{silentpastattractor}
The past asymptotic limit in a regime where $H>0$ resides on
the type $I-VII$ part of the silent boundary if the strong
energy condition is asymptotically fulfilled.
\end{proposition}

\begin{proof}
We have
\be\label{detNeq}
\parb_0\, \mathrm{det}(N_{\alpha\beta}) = 3q\,
\mathrm{det}(N_{\alpha\beta}) , \ee
on the silent boundary, where
\be q = 2\Sigma^2 + \textfrac{1}{2}(\Omega_\mathrm{tot} +
3P_\mathrm{tot}) .\ee
%
If the strong energy condition $\Omega_\mathrm{tot} +
3P_\mathrm{tot}\geq 0$ holds,\footnote{It is likely that there
exist generic solutions with only a positive cosmological
constant as source (with $\Omega_\mathrm{tot} +
3P_\mathrm{tot}=-2\Omega_\mathrm{tot}< 0$) that asymptotically
behaves as generic vacuum solutions with asymptotically silent
and local past singularities, and hence it should be possible
to relax the condition $\Omega_\mathrm{tot} +
3P_\mathrm{tot}\geq 0$, but for simplicity we refrain from
doing this.} then $q\geq0$, and $q=0$ only when
$\Omega_\mathrm{tot}+ 3P_\mathrm{tot}=0$ and $\Sigma^2=0$, but
then
\begin{equation}
\parb^2_0 \, {\rm det}(N_{\alpha\beta})|_{q=0} = 0, \quad
\parb^3_0 \, {\rm det}(N_{\alpha\beta})|_{q=0} = 2 \,[{}^3\!{\cal S}^{\gamma\delta}
\,{}^3\!{\cal S}_{\gamma\delta}] \, {\rm det}(N_{\alpha\beta}),
\end{equation}
where ${}^3\!{\cal S}^{\gamma\delta} \,{}^3\!{\cal
S}_{\gamma\delta}>0$ when ${\rm det}(N_{\alpha\beta})\neq 0$;
it follows that
\be\label{detnzero} {\rm det}(N_{\alpha\beta}) \rightarrow 0
\ee
toward the past singularity. Thus the past asymptotic limit of
the dynamics must reside on the
$\mathrm{det}(N_{\alpha\beta})=0$ subset, i.e., the Bianchi
type I-VII part of the silent boundary.
\end{proof}

\begin{corollary}\label{sigmabound}
\begin{displaymath}
\lim_{x^0 \to -\infty} |\Sigma_{\alpha\beta}| \leq 2.
\end{displaymath}
\end{corollary}

\begin{proof}
$\mathrm{det}(N_{\alpha\beta})=0$ implies that $\Omega_k\geq
0$, which, together with the Gauss constraint
$1-\Sigma^2=\Omega_k + \Omega\geq  0$, yields
\begin{equation}
\Sigma^2\leq 1 \qquad \Rightarrow \qquad -2\leq
\Sigma_{\alpha\beta}\leq 2.
\end{equation}
\end{proof}

On the silent boundary
\be
\parb_0\, A^2 = 2(q\delta_\alpha{}^\beta -
\Sigma_\alpha{}^\beta)A^\alpha\,A_\beta , \ee
and hence, assuming the validity of `the locality conjecture'
(1), and thereby that Eq. \eqref{positivity condition} holds,
\be A_\alpha \rightarrow 0 \ee
toward the singularity, i.e., the past attractor has to reside
on the subset that consists of the union of the class A
($A_\alpha=0$) type I, II, VI$_0$, and VII$_0$ subsets on the
silent boundary.

On the class A part of the silent boundary~\eqref{pec} reduces
to
\be \label{veqclassA}
\parb_0 v = \bar{G}_-^{-1}\,(1-v^2)\,
(3c_s^2 - 1 - \Sigma_{\alpha\beta}\,c^\alpha c^\beta )\, v. \ee
Corollary~\eqref{sigmabound} and Eq.~\eqref{veqclassA} indicate
that there is a bifurcation in the dynamics of the particular
velocities of the fluids when $c_s^2=1$. We will therefore
below distinguish between three main cases, based on the
asymptotic properties of the equations of state:
\begin{itemize}
\item[(i)] There exists at least one fluid with an
    asymptotically ultra-stiff equation of state, i.e.,
    $c_s^2>1, w>1$ when $x^0 \rightarrow -\infty$.
\item[(ii)] All perfect fluids have asymptotic equations of
    state such that $c_s^2<1$, $w<1$ when $x^0 \rightarrow
    -\infty$, except for at least one fluid which has an
    asymptotically stiff equation of state, i.e., $c_s=1,
    w=1$ when $x^0 \rightarrow -\infty$.
\item[(iii)] All perfect fluids have asymptotic equations
    of state such that $c_s^2<1$, $w<1$ when $x^0
    \rightarrow -\infty$, i.e., all equations of state are
    softer than a stiff equation of state asymptotically.
\end{itemize}
We will denote the three cases as the (asymptotically)
\textit{ultra-stiff\/}, \textit{stiff\/}, and \textit{soft\/}
cases, respectively; as we will see, their past dynamics is
associated with an increasingly complicated and challenging
analysis.\footnote{The physical status of an ultra-stiff
equation of state can be questioned since $c_s$ is larger than
the speed of light, however, it is of interest for structural
stability reasons to study sources with fluids with such an
equation of state, moreover, in~\cite{collim05}, and references
therein, the study of problems associated with ultra-stiff
equations of state is motivated by considering broader
theoretical contexts than general relativity.}

To proceed we prove, under assumption~\eqref{positivity
condition}, the following lemma:

\begin{lemma}\label{Non-Taub condition}
\begin{equation}
\int_{\tilde{x}^0}^{-\infty}{\rm eig} (2\delta_\alpha{}^\beta -
\Sigma_\alpha{}^\beta){\cal M}dx^0 = -\infty.
\end{equation}
\end{lemma}

\begin{proof}
Case (i): As shown next, $\Sigma_{\alpha\beta} \to 0$ in case
(i), and hence the integral diverges. Case (ii) and (iii):
Eq.~\eqref{pfrel} yields $\Omega_{tot} \geq P_{tot}$, which
leads to the inequality $q = 2\Sigma^{2} +
\textfrac{1}{2}(\Omega_{tot} + 3P_{tot}) = 2 -
\textfrac{3}{2}(\Omega_{tot} - P_{tot}) -2\Omega_{k} \leq 2$.
This combined with Eq.~\eqref{positivity condition} gives
\begin{eqnarray}
&&\int_{\tilde{x}^0 }^{-\infty} {\rm eig}(2\delta_\alpha{}^\beta -
\Sigma_\alpha{}^\beta){\cal M}dx^0 = \int_{\tilde{x}^0}
^{-\infty}{\rm eig}\left[(2-q)\delta_\alpha{}^\beta + (q \delta_\alpha
{}^\beta - \Sigma_\alpha{}^\beta)\right]{\cal M}dx^0 \leq \nonumber\\
&&\int_{\tilde{x}^0}^{-\infty}{\rm eig}(q \delta_\alpha{}^\beta -
\Sigma_\alpha{}^\beta){\cal M}dx^0 = -\infty.
\end{eqnarray}
\end{proof}

\begin{proposition}\label{case1}
The past asymptotic state in case (i) is characterized by
\begin{equation}
\Omega_{\mathrm{ultra-stiff}} \to 1; \qquad (\Sigma_{\alpha\beta}
,\, N_{\alpha\beta},\, v_{\mathrm{ultra-stiff}},\, \Omega_{(i)})
\to 0,\ \forall i \neq {\scriptstyle \mathrm{ultra-stiff}}.
\end{equation}
\end{proposition}

\begin{proof}
Corollary~\eqref{sigmabound} and Eq.~\eqref{veqclassA} give
that $\lim_{x^0 \to -\infty} v_{\mathrm{ultra-stiff}}=0$. On
the class A $v_{\mathrm{ultra-stiff}}=0$ boundary
Eq.~\eqref{ompfsil} yields
\begin{align}\label{rhocomp}
&
\parb_0\, {\rm ln} (\Omega_{(i)}/\Omega_{\text{ultra-stiff}}) = 3(w_{\text{ultra-stiff}}-1)\nonumber\\
& \qquad  + (G_+)_{(i)}^{-1}[3(1-w_{(i)})(1-v_{(i)}^2) +
(1+w_{(i)})(2\delta_{\alpha\beta} - \Sigma_{\alpha\beta})v^\alpha_{(i)}
v^\beta_{(i)} ] > 0,\end{align}
where the $(i)$:th fluid has a comparably asymptotic soft
equation of state. For simplicity we have assumed that the
ultrastiff fluid obeys an asymptotically linear ultra-stiff
equation of state such that
$w_{\text{ultra-stiff}}=\lim_{\tilde{\rho}_{\text{ultra-stiff}}\rightarrow\infty}
(w)$; in the case of several fluids with the same asymptotic
ultra-stiff asymptotic equation of state,
$\Omega_{\text{ultra-stiff}}$, represent their total
contributions. Since the r.h.s. of~\eqref{rhocomp} is strictly
positive it follows that
$\Omega_{(i)}/\Omega_{\text{ultra-stiff}}\rightarrow 0$ toward
the past, and since $\Omega_{\text{ultra-stiff}}$ is bounded,
because of the Gauss constraint $1-\Sigma^2 - \Omega_k -
\Omega_\mathrm{tot}=0$ and the non-negativity of the energy
densities and $\Omega_k$, this leads to that the ultra-stiff
fluid(s) dominates toward the singularity, and hence
$\Omega_{(i)} \rightarrow 0$; thus the attractor in the
ultra-stiff case (i) resides on the class A Bianchi type I --
VII$_0$ part of the silent boundary with
$v_{\text{ultra-stiff}}=0$, $\Omega_{(i)}= 0$, for all $i$
except for the $i$ associated with the ultra-stiff fluid(s),
subset. This leads to that~\eqref{ompfsil} asymptotically
yields
\be\label{omultra-stiff}
\parb_0 \Omega_{\text{ultra-stiff}} =
-[3(w_{\text{ultra-stiff}}-1)(1-\Omega_{\text{ultra-stiff}}) +
4\Omega_k]\,\Omega_{\text{ultra-stiff}} , \ee
and hence, due to that $\Omega_{\text{ultra-stiff}}\leq 1$,
asymptotically $\Omega_{\text{ultra-stiff}}=1$ and
$\Omega_k=0$, and thus, because of the Gauss constraint,
$\Sigma^2=0$. That $\Omega_k=0$ and $\Sigma^2=0$ yield that the
past attractor in the ultra-stiff case must reside on the
isotropic type I subset or the isotropic type VII$_0$ subset;
in the latter case we can choose a Fermi frame in which
$N_{\alpha\beta}=\text{diag}(0,N,N)$, or cycle, which yields
$\parb_0\,N = q N = \frac{1}{2}(1 + 3w_{\text{ultra-stiff}})N$,
and hence $N\rightarrow 0$, i.e., the past attractor is located
on the isotropic type I subset, which is a frame independent
statement; we will refer to the silent isotropic type I subset
as the silent \textit{Friedmann subset} ${\cal F}$.

The above arguments are easily generalized to the situation
when the most ultra-stiff equation(s) of state does not have a
limit, but a lower bound $w_{\text{ultra-stiff}}^->1$; one
still obtains that the past attractor resides on ${\cal F}$
with $\Omega_{\text{ultra-stiff}}=1, \,\Omega_{(i)}=0, \,
v_{\text{ultra-stiff}}=0, \, \Sigma^2=0$, even though $q$ has
no limit.
\end{proof}

\begin{proposition}\label{case2}
The past asymptotic state in case (ii) is characterized by
\begin{equation}
q \to 2,\qquad (N_{\alpha\beta},\, v_{\mathrm{stiff}},\, \Omega_{(i)})
\to 0,\ \forall i \neq {\scriptstyle \mathrm{stiff}}.
\end{equation}
\end{proposition}

\begin{proof}
The analysis of the stiff case (ii) proceeds with similar
arguments as in the proof of case (i), but with the extra
condition of lemma \ref{Non-Taub condition}. This leads to that
$v_\mathrm{stiff} = 0$ asymptotically, and that
$\Omega_{(i)}=0$ asymptotically for all fluids with equations
of state that are asymptotically softer than the asymptotically
stiff fluid(s). Hence the past asymptotic state resides on the
union of the class A Bianchi type I, II, VI$_0$, VII$_0$
subsets for a single orthogonal stiff fluid, where
$c^\alpha_{(i)}$ and $v_{(i)}$ act as test fields, i.e., fields
that do not affect the spacetime geometry but are affected by
it. The past asymptotic dynamics for the single orthogonal
stiff fluid case in Bianchi types I, II, VI$_0$, VII$_0$ is
well known~\cite{rin01,waiell97}, and from this it follows that
the past attractor resides on the type I subset where
$\Omega_{\text{stiff}}=\hat{\Omega}_{\text{stiff}},\, q=2$,
where we have introduced the convention of using hats on purely
spatially dependent, i.e., temporally constant, quantities. It
therefore follows that the past attractor in case (ii) resides
on the type I subset where
\be
\Omega_\mathrm{tot}=\Omega_{\text{stiff}}=\hat{\Omega}_{\text{stiff}},\qquad
\Omega_{(i)}=0 ,\qquad v_{\text{stiff}}=0 ,\qquad q=2 ; \ee
we will refer to this subset as the silent \textit{Jacobs
subset\/} ${\cal J}$ (the exact solutions for a single stiff
perfect fluid in Bianchi type I were first found by
Jacobs~\cite{jac68}).
\end{proof}

We now turn to the behavior of $\Omega_{(i)}$ in the soft case
(iii). For this case we have no proof, but we expect that the
`matter does not matter' conjecture (2) holds, and that
$\Omega_{(i)}\rightarrow 0$ for all $i$ toward the past
singularity, and that the past attractor hence resides on the
vacuum subset $\Omega_\mathrm{tot}=0$. The reason for the
expectation that $\Omega_\mathrm{tot}=0$ asymptotically is that
there exists evidence for that this happens when one has one
fluid with a soft equation of
state~\cite{rin01,heiugg09a,uggetal03}, and it seems reasonable
that one can apply this result for each fluid individually;
furthermore, in the vacuum case there exists evidence that the
past attractor resides on the union of the silent vacuum type I
subset, known as the silent \textit{Kasner subset} ${\cal K}$,
and the silent vacuum type II
subset~\cite{andetal05,rin01,heiugg09a,uggetal03,heietal09}.
Moreover, in two previous studies of tilted multi-fluid models
of Bianchi type I~\cite{sanugg08,san09} we presented evidence
that indicated that the past attractor of the Bianchi type I
models with two soft fluids resided on ${\cal K}$, and since we
expect that ${\cal K}$ plays a `dominant' role in the
asymptotic dynamics this gives further support for the claim
that $\Omega_\mathrm{tot}\rightarrow 0$.

From Proposition \ref{case1} in case (i), Proposition
\ref{case2} in case (ii), and the `matter does not matter'
conjecture (2) in case (iii), it follows that asymptotically
toward the past
\begin{equation}
Q^\alpha_\mathrm{tot}=0 \qquad \text{and} \qquad
\Pi^{\alpha\beta}_\mathrm{tot}=0,
\end{equation}
in all cases, since $\Omega_{(i)}=0$, for all $i$, except for
the asymptotically `dominant' ultra-stiff fluid(s) in case (i)
and the asymptotically stiff fluid(s) in case (ii), but in
those cases $v_\mathrm{ultra-stiff}=0$ and
$v_\mathrm{stiff}=0$, respectively.

For all Class A models with $Q^\alpha_\mathrm{tot}=0$ and
$\Pi^{\alpha\beta}_\mathrm{tot}=0$ it is possible to
simultaneously diagonalize $N_{\alpha\beta}$ and
$\Sigma_{\alpha\beta}$ in a Fermi frame. The reason for this is
as follows: In class A $Q^\alpha_\mathrm{tot}=0$ leads to that
the Codazzi constraint~\eqref{dCodazzisil} takes the form
$\epsilon_{\alpha\delta}{}^{\gamma}\,N^\delta{}_\beta\,\Sigma^\beta{}_\gamma=0$,
which implies that $N_{\alpha\beta}$ and $\Sigma_{\alpha\beta}$
are simultaneously diagonalizable for a given arbitrary value
of $x^0$ ($N_{\alpha\beta}$ transforms as a tensor density on
the silent boundary under spatial frame rotations).
Furthermore, the preservation of the simultaneous
diagonalization during evolution is possible because
$\Pi^{\alpha\beta}_\mathrm{tot}=0$, but it also requires that
one uses a Fermi frame. We hence expect that it is possible to
asymptotically diagonalize $\Sigma_{\alpha\beta}$ and
$N_{\alpha\beta}$ in a frame that is asymptotically a Fermi
frame\footnote{At least generically, the results
in~\cite{heietal09} show that some degrees of freedom only seem
to be statistically suppressed, and that there may be a few
timelines with different behavior; hence there may be some
timelines with different asymptotic dynamics than that we
presently describe.} so that
\begin{equation}
R_\alpha = 0, \quad
\Sigma_{\alpha\beta} = \mathrm{diag}(\Sigma_1,\Sigma_2,\Sigma_3), \quad
N_{\alpha\beta} = \mathrm{diag}(N_1,N_2,N_3); \quad \Sigma_1 + \Sigma_2 + \Sigma_3 = 0 .
\end{equation}

In all fluid cases, the silent Bianchi type I subset plays a
prominent role, indeed, according to the previous analysis the
past attractors for the ultra-stiff and stiff cases reside
there, and we therefore now turn to this subset in more detail.

\subsection{The silent Bianchi type I
subset}\label{subsec:typeI}

As follows from the previous subsection for cases (i) and (ii),
and as conjectured for case (iii), sources that consist of
multiple perfect fluids lead to that the past asymptotic subset
for Bianchi type I resides on the subset with
$\Omega_{(i)}=0\,\forall i$, except for the `dominant' matter
component(s) $\Omega_\mathrm{ultra-stiff}=1$ (where also
$v^\alpha_\mathrm{ultra-stiff}=0$) and $\Omega_\mathrm{stiff} =
\hat{\Omega}_\mathrm{stiff}$ (where also
$v^\alpha_\mathrm{stiff}=0$) in cases (i) and (ii),
respectively. This implies that the past asymptotic dynamics on
the silent type I boundary resides on ${\cal F}$, ${\cal J}$,
and ${\cal K}$, for the ultra-stiff, stiff, and soft cases,
respectively, where ${\cal K}$ constitutes the boundary of
${\cal J}$ in the stiff case. In $\Sigma_{\alpha\beta}$-space
the field equations for Bianchi type I immediately lead to that
these subsets are characterized by the following eigenvalues
for $\Sigma_{\alpha\beta}$:
\begin{subequations}
\begin{align}
{\cal F}:&\quad \Sigma_\alpha = \hat{\Sigma}_\alpha
    = 0,\ \forall \ \alpha, \quad \Sigma^2=0\,\, \Leftrightarrow
    \,\, \Omega_\mathrm{tot} = 1.\\
{\cal J}:&\quad \Sigma_\alpha =
    \hat{\Sigma}_\alpha,\ \forall \ \alpha,\qquad\,\,\,\,\,
    \Sigma^2 = \hat{\Sigma}^2 = 1 - \Omega_\mathrm{tot} =
    1 - \hat{\Omega}_\mathrm{stiff}.\\
{\cal K}:&\quad \Sigma_\alpha =
    \hat{\Sigma}_\alpha,\ \forall \ \alpha,\qquad\,\,\,\,\,
    \Sigma^2 = \hat{\Sigma}^2 = 1 \,\, \Leftrightarrow\,\,
    \Omega_\mathrm{tot} = 0.
\end{align}
\end{subequations}

The eigenvalues $\Sigma_\alpha=\hat{\Sigma}_\alpha$ can be
expressed in terms of the \textit{shape parameters} $p_\alpha$,
see~\cite{limetal06}, defined according to
\begin{equation}\label{shapedef}
(\hat{\Sigma}_1,\hat{\Sigma}_2,\hat{\Sigma}_3) =
(3p_1-1,3p_2-1,3p_3-1) , \qquad
p_1+p_2+p_3 = 1 ,
\end{equation}
where we have omitted the hats on the spatially dependent
$p_\alpha$ to conform with standard notation. In terms of the
shape parameters, the past asymptotic states on the silent
Bianchi type I subset for the three subsets are described by:
\begin{subequations}
\begin{align}
{\cal F}: &\quad (p_1,p_2,p_3)\ =\
    \textfrac{1}{3}\,(1,1,1).\\
{\cal J}:&\quad p_1^2 + p_2^2 + p_3^2 \ = \ 1 -
    \textfrac{2}{3}\hat{\Omega}_{\text{stiff}} < \ 1.
   \label{shapecond}\\
{\cal K}:&\quad p_1^2 + p_2^2 + p_3^2\ =\ 1.
\end{align}
\end{subequations}

Even though $\Omega_{(i)}=0\,\,\forall i$ (with the exception
of $\Omega_\mathrm{ultra-stiff}=1$ and $\Omega_\mathrm{stiff} =
\hat{\Omega}_\mathrm{stiff}$ in cases (i) and (ii),
respectively), the subsets ${\cal F}$, ${\cal J}$, and ${\cal
K}$ also involve the equations for $c^\alpha_{(i)}$ and
$v_{(i)}$, which act as test fields, i.e., fields that do not
affect the spacetime geometry but are affected by it, and hence
a complete past asymptotic description also involves the
asymptotic determination of these fields. For this purpose we
use a shear diagonalized Fermi frame so that
$\Sigma_{\alpha\beta}=\mathrm{diag}(3p_1-1,3p_2-1,3p_3-1)$ and
$R_\alpha=0$, and insert $A_\alpha=0,\, N_{\alpha\beta}=0$,
which characterizes Bianchi type I, into the
equations~\eqref{cpec} for $c_\alpha$; this leads to:
\begin{subequations}\label{pecdir}
\begin{align}
\parb_0\, c_1 &= 3[(p_2 - p_1)c_2^2 +
(p_3 - p_1)c_3^2]\,c_1 ,\\
\parb_0\, c_2 &= 3[(p_3 - p_2)c_3^2 +
(p_1 - p_2)c_1^2]\,c_2 ,\\
\parb_0\, c_3 &= 3[(p_1 - p_3)c_1^2 +
(p_2 - p_3)c_2^2]\,c_3 ,
\end{align}
\end{subequations}
where we again for simplicity have dropped the index $(i)$.
These equations, which decouple from the equation for $v$, can
be treated as a separate dynamical system that satisfies the
constraint $c_\alpha c^\alpha=1$, i.e., we have a dynamical
system on a sphere with unit radius, parameterized by $p_1$,
$p_2$, and $p_3$. We note that this system is the same as that
for $v^\alpha$ when $v^2=1$, i.e., the dynamics for $c_\alpha$
is the same as for the extreme tilt subset $v^2=1$, which
in~\cite{uggetal03} was examined by means of spherical
coordinates in the case $p_1<p_2<p_3$.

\begin{proposition}\label{pecdirprop}
The past asymptotic state of the system \eqref{pecdir} is given
by:
\begin{itemize}
\item[]Case (i): $c_\alpha=\hat{c}_\alpha$.
\item[] Cases (ii) and (iii): Let $(\alpha\beta\gamma) =
    (123)$, or a permutation thereof. (a) If $p_\alpha\leq
    p_\beta<p_\gamma$, then $c_\alpha,\,c_\beta \to 0,\quad
    c_\gamma \to \pm1$ when $c_\gamma \gtrless 0$. (b) If
    $p_\alpha < p_\beta=p_\gamma$, then $c_\alpha \to 0,\
    c_\beta, c_\gamma \to \hat{c}_\beta,\hat{c}_\gamma,
    \hat{c}_\beta^2 + \hat{c}_\gamma^2=1$.
\end{itemize}
\end{proposition}

\begin{proof}
In case (i) where $(p_1,p_2,p_3)=\frac{1}{3}(1,1,1)$, it
follows directly that $\parb_0 c_\alpha = 0$ and hence
$c_\alpha=\hat{c}_\alpha$.

Cases (ii) and (iii) can be treated collectively. We first note
that if $p_\alpha<p_\beta<p_\gamma$, where $(\alpha\beta\gamma)
= (123)$, or a permutation thereof, then the
system~\eqref{pecdir} admits the invariant subsets ${\cal
C}_{12}$ on which $c_3=0$, and cycle, leading to a division of
the sphere into six disjoint subsets with the subset ${\cal
C}_{12},\,{\cal C}_{23},\,{\cal C}_{31}$ as boundaries,
furthermore, the intersections of these subsets yield the fix
points $\mathrm{C}^\pm_{\alpha}$ for which $c_\alpha=\pm
1,\,c_\beta=c_\gamma=0$, $(\alpha\beta\gamma)=(123)$, and
cycle. If $p_\alpha=p_\beta\neq p_\gamma$, where
$(\alpha\beta\gamma) = (123)$, and cycle, then the
system~\eqref{pecdir} also admits subsets when one of the
components $c_1,\,c_2,$ or $c_3$ is zero, but the subset ${\cal
C}_{\alpha\beta}$ on which $c_\gamma=0$ reduces to a circle of
fix points with
$c_\alpha=\hat{c}_\alpha,\,c_\beta=\hat{c}_\beta,\,\hat{c}_\alpha^2+\hat{c}_\beta^2=1$,
which we denote by $\mathrm{C}_{\alpha\beta}^\ocircle$.

If $p_\alpha\leq p_\beta<p_\gamma$, then $\parb_0\, c^2_\alpha
> 0$, and hence $c_\alpha \to 0$. This reduces the system to
the subset ${\cal C}_{\beta\gamma}$ where $p_\beta < p_\gamma$
gives $\parb_0\, c^2_\beta > 0$, $\parb_0\, c_\gamma =
f(c_\beta)\, c_\gamma$, where $f(c_\beta) < 0$, and hence
$c_\beta \to 0$ and $c_\gamma \to \pm 1$ toward the past when
$c_\gamma \gtrless 0$. If $p_\alpha < p_\beta = p_\gamma$, then
$c_\alpha \to 0$ still holds, but in this case the system
reduces to the circle of fix points
$\mathrm{C}_{\beta\gamma}^\ocircle$.
\end{proof}

\begin{corollary}\label{velocityshearalignment}
The past asymptotic peculiar velocity direction(s) $c^\alpha$
of the test fields coincide with the asymptotic eigenvector(s)
of $\Sigma_{\alpha\beta}$ associated with the eigenvalue(s)
$\hat{\Sigma}_\mathrm{max} = \max(\hat{\Sigma}_1,
\hat{\Sigma}_2, \hat{\Sigma}_3) = \max(3p_1-1,3p_2-1,3p_3-1)$.
\end{corollary}

\begin{proof} This follows immediately from the proof
of~\eqref{pecdirprop}.
\end{proof}

We now turn to the past asymptotic behavior for the peculiar
test speeds $v$ on the type I subset. By regarding $c_\alpha$
as time-dependent coefficients in the evolution equation for
$v$, we can apply a theorem by Strauss and
Yorke~\cite{stryor67} that implies that $v$ is past
asymptotically determined by the past asymptotics of
$c_\alpha$. Corollary \ref{velocityshearalignment} then reduces
Eq.~\eqref{veqclassA} for $v$ to
\begin{equation}\label{v2sub}
\parb_0 v = 3\bar{G}_-^{-1}(1-v^2)(c_s^2 - p_\mathrm{max})\,v ,
\qquad \text{where} \qquad
p_\mathrm{max} = \max(p_1,p_2,p_3) .
\end{equation}
Consequently $v$ is monotonically decreasing (increasing)
toward the past if $c_s^2
> p_\mathrm{max} = \frac{1}{3}(1+ \hat{\Sigma}_\mathrm{max})$ ($c_s^2 <
p_\mathrm{max} = \frac{1}{3}(1 + \hat{\Sigma}_\mathrm{max})$),
and hence $v=0$ ($v=1$), while $v=\hat{v}$ if
$c_s^2=p_\mathrm{max}$, asymptotically toward the past; in
these formulas $c_s^2$ refers to the asymptotic limit of
$c_s^2$ when $\tilde{\rho}\rightarrow\infty$ (for simplicity we
assume that $c_s^2$ has such a limit, however, many of our
results are easily generalized to the case when $c_s^2$ has
asymptotic bounds, but no limit).

\begin{itemize}
\item[] Case (i): $p_\mathrm{max}=\frac{1}{3}$ and hence
    $v\rightarrow 0$ when $c_s^2>\frac{1}{3}$;
    $v\rightarrow \hat{v}$ when $c_s^2=\frac{1}{3}$;
    $v\rightarrow 1$ when $c_s^2<\frac{1}{3}$.
\item[] Case (ii): Eqs.~\eqref{shapedef}
    and~\eqref{shapecond} yield that
    $\frac{1}{3}(1+\hat{\Sigma})\leq p_\mathrm{max} \leq
    \frac{1}{3}(1+2\hat{\Sigma})<1$, where
    $\hat{\Sigma}=\sqrt{1-\hat{\Omega}_{\text{stiff}}}<1$.
    Hence $c_s^2 < \frac{1}{3}\Rightarrow v\rightarrow 1$
    toward the past; if $c_s^2 > \frac{1}{3}$ there exist
    some $p_\mathrm{max}$ values on ${\cal J}$ for which
    $v\rightarrow 0$, some for which $v\rightarrow
    \hat{v}$, and some for which $v\rightarrow 1$,
    depending on if $c_s^2>p_\mathrm{max}$,
    $c_s^2=p_\mathrm{max}$, or $c_s^2<p_\mathrm{max}$ (the
    smallest possible $p_\mathrm{max}$ value is
    $\frac{1}{3}$ and occurs when $\Sigma^2=0$).
\item[] Case (iii): Eqs.~\eqref{shapedef}
    and~\eqref{shapecond} yield that $\frac{2}{3}\leq
    p_\mathrm{max} <1$.\footnote{We exclude that
    $p_\mathrm{max}=1$; this is intimately connected with
    that the assumption~\eqref{positivity condition} holds.
    Note that $p_\mathrm{max}=1$ has been proved to be
    excluded for the non-LRS Bianchi types VIII and IX
    cases~\cite{rin01,heiugg09a,rin00}.} Hence $c_s^2 <
    \frac{2}{3} \Rightarrow v\rightarrow 1$ toward the
    past. If $c_s^2 > \frac{2}{3}$ there exist some points
    on ${\cal K}$ for which $v\rightarrow 0$ and some for
    which $v\rightarrow 1$.
\end{itemize}

It is of interest to note that $c_s^2\rightarrow 1 \Rightarrow
v\rightarrow 0$ everywhere on ${\cal J}$ and ${\cal
K}$.\footnote{Except for the excluded points with
$p_\mathrm{max}=1$.}

\subsection{Stability and instability of the type I subset on the silent boundary}\label{perturbI}

For the ultra-stiff case (i) it is easily seen that ${\cal F}$
is a stable subset w.r.t. perturbations of $E_\alpha{}^i$,
$A_\alpha$, $N_{\alpha\beta}$, $\Sigma_{\alpha\beta}$,
$\Omega_{(i)}$, and hence there exists a past attractor in the
full state space that resides on ${\cal F}$ in this case, a
statement that is also supported by the analysis
in~\cite{collim05}. We therefore turn to the past attractor for
the stiff and soft cases (ii) and (iii), respectively.

To identify the past attractor subset on the silent boundary we
next linearly perturb ${\cal J}$ and ${\cal K}$ by using a
Fermi frame in which the perturbed Bianchi type I subsets are
expressed in a Fermi frame with diagonalized shear
$\Sigma_{\alpha\beta}=\mathrm{diag}
(\hat{\Sigma}_1=3p_1-1,\hat{\Sigma}_2=3p_2-1,\hat{\Sigma}_3=3p_3-1)$:
\begin{subequations}\label{geomstabeq}
\begin{align}
A_\alpha^{-1}\,\parb_0\,A_\alpha|_{{\cal J},{\cal K}} &= 2-\hat{\Sigma}_{\alpha} = 3(1-p_\alpha),
\label{Astab}\\
N_{\alpha}^{-1}\,\parb_0\, N_{\alpha}|_{{\cal J},{\cal K}} &=2(1+\hat{\Sigma}_\alpha) = 6p_\alpha \qquad
\quad\, \text{where} \qquad\, N_\alpha=N_{\alpha\alpha},
\label{Ndiastab}\\
N_{\alpha\beta}^{-1}\,\parb_0\, N_{\alpha\beta}|_{{\cal J},{\cal K}} &= 2-\hat{\Sigma}_\gamma  = 3(1-
p_\gamma)\qquad \text{where}\quad (\alpha\beta\gamma) = (123)\,\,
\text{and cycle},\label{Noffdstab}\\
\Omega^{-1}_{(i)}\parb_0\,\Omega_{(i)}|_{{\cal J},{\cal K}} &= 3G_+^{-1}\left[(1-w)(1-v^2) +
(1+w)(1 - p_{\mathrm{max}})\,v^2\right],\label{Omstab}
\end{align}
\end{subequations}
where the above equations refer to separate components. The
notation $|_{{\cal J},{\cal K}}$ indicates evaluation at ${\cal
J}$ in the stiff case (ii), and at ${\cal K}$ in the soft case
(iii). Note that we have obtained the same form for the
equations in the stiff case (ii) and the soft case (iii), since
$q=2$ in both cases. Moreover, in the case of $\Omega_{(i)}$ we
have in addition inserted the type I past attractor value
associated with ${\cal J}$, ${\cal K}$ for
$\Sigma_{\alpha\beta}\,c^\alpha\,c^\beta$ according to
Corollary \ref{velocityshearalignment}.

Notably there are no equations for $\Sigma_{\alpha\beta}$
in~\eqref{geomstabeq}. The reason for this is that the
stability analysis of $\Sigma_{\alpha\beta}$ depends on the
choice of spatial frame. However, Eqs.~\eqref{geomstabeq} hold
for any spatial frame that admits
$\Sigma_{\alpha\beta}=\mathrm{diag}(\hat{\Sigma}_1,\hat{\Sigma}_2,\hat{\Sigma}_3)$
and $R_\alpha=0$ as an invariant subset on ${\cal J}$ and
${\cal K}$; furthermore, when projecting out the peculiar
velocities, which we will do in the reminder of this
subsection, these sets form sets of fix points for the
projected system of equations. Nevertheless, in the full state
space, as well as on the silent boundary where in general
$A_\alpha\neq 0, Q_\alpha \neq 0, \Pi_{\alpha\beta}\neq 0$, the
shear cannot be diagonalized in a Fermi frame, and thus when we
consider the general dynamics we cannot assume that
$\Sigma_{\alpha\beta}=\mathrm{diag}(\hat{\Sigma}_1,\hat{\Sigma}_2,\hat{\Sigma}_3)$
and $R_\alpha=0$. For our purposes, however, it suffices to
consider an asymptotic spatial frame choice. In this paper we
will use an asymptotic Fermi frame, i.e.,
$R_\alpha=0$,\footnote{There exists other interesting choices,
e.g. $R_\alpha = \epsilon_\alpha\,\Sigma_{\beta\gamma}$, where
$(\alpha\beta\gamma)=(123)$, or cycle, and where
$\epsilon_\alpha$ is equal to $\pm 1$;
$\epsilon_\alpha=(-1,1,-1)$ is connected with the Iwasawa frame
used in e.g.~\cite{heietal09}, $\epsilon_\alpha=(1,1,1)$ is the
frame choice used in~\cite{uggetal03}. For these choices
$R_\alpha$ destabilizes parts of the fix point sets on ${\cal
J},\,{\cal K}$ by inducing so-called frame transitions (also
known as centrifugal bounces in a Hamiltonian context,
see~\cite{heietal09}), trajectories that connect one fix point
representation of a type I solution with another, by means of
an axes permutation~\cite{heietal09}.} but note that we cannot
e.g. diagonalize the shear in such a frame, except
asymptotically. On the Kasner subset ${\cal K}$, and the Jacobs
subset ${\cal J}$, a Fermi frame choice leads to that the
$\Sigma_{\alpha\beta}$ evolution equation~\eqref{sigsil}
immediately yields
\be \Sigma_{\alpha\beta}=\hat{\Sigma}_{\alpha\beta},
\ee
since $q=2$ in both cases. Thus, because of that
$\Sigma^2=\hat{\Sigma}^2 = 1- \hat{\Omega}_\mathrm{stiff}$
there exists a ellipsoidal ball (ellipsoid) of fix points in
the (shear projected) ${\cal J}$ (${\cal K}$) case, which
corresponds to a center manifold. However, in this case a
temporally constant rotation of axes that diagonalizes
$\Sigma_{\alpha\beta}$ so that $\Sigma_{\alpha\beta} =
\text{diag}(\hat{\Sigma}_1,\hat{\Sigma}_1,\hat{\Sigma}_1)$,
which leads to~\eqref{geomstabeq}.

From~\eqref{geomstabeq} it follows that $A_\alpha$,
$N_{\alpha\beta}$, when $\alpha\neq\beta$, and $\Omega_{(i)}$
($(i)\neq \mathrm{stiff}$) are stable toward the past
everywhere on ${\cal J}$ and ${\cal K}$, with the exception of
the non-transversally-hyperbolic so-called Taub points on
${\cal K}$, where $(p_1,p_2,p_3)= (1,0,0)$, and cycle (or
equivalently $(\hat{\Sigma}_1,\hat{\Sigma}_2,\hat{\Sigma}_3)=
(2,-1,-1)$, and cycle). However, we have shown that if
lemma~\ref{Non-Taub condition} holds, which depends on that the
condition~\eqref{positivity condition} holds, then $A_\alpha$
and $\Omega_{(i)}$ both tend to zero toward the singularity and
thus the states $A_\alpha=0$ and $\Omega_{(i)}=0$ are past
stable, even though they are not linearly stable
everywhere.\footnote{There hence exists an intricate connection
between avoidance of the Taub points, which in turn are part of
the Taub subset described in Appendix~\ref{invbound}, and
asymptotic locality via lemma~\ref{Non-Taub condition} and
condition~\eqref{positivity condition}. Determining exactly
what this connection is poses a formidable and important
challenge. In this context it is worth mentioning that there
may exist an open set of solutions with so-called weak null
singularities, which are not asymptotically silent or local;
moreover, these singularities seem to be intimately associated
with the Taub subset~\cite{limetal06}.} The decoupling of the
$\Omega_{(i)}$ equations~\eqref{Omstab} from each other, and
their shared linear stability properties in conjunction with
the previous non-linear stability result, gives some support
for the `matter does not matter' conjecture (2), in the context
of that the condition~\eqref{positivity condition} holds, cf.
also~\cite{rin01,heiugg09b,heietal09}.

The stability toward the past of $N_\alpha$ depends on the sign
of $p_\alpha=\frac{1}{3}(1+\hat{\Sigma}_\alpha)$. In the stiff
case (ii) it follows from~\eqref{geomstabeq} that the part of
${\cal J}$ that obeys $\hat{\Sigma}_\alpha=-1\,\forall \alpha$,
i.e. with $\hat{\Sigma}_\alpha >-1$ or, equivalently,
$p_\alpha>0$, is stable w.r.t. $N_\alpha$ perturbations toward
the past; we will denote this part of ${\cal J}$ as ${\cal
J}^{\Delta}$. Outside ${\cal J}^{\Delta}$, ${\cal J}$ have an
unstable mode associated with $N_\alpha$ when
$\hat{\Sigma}_\alpha<-1$, or, equivalently, when $p_\alpha<0$.
This follows from that only one of $p_1,p_2$, and $p_3$ is
negative, because $p_1+p_2+p_3=1$ and
$p_1^2+p_2^2+p_3^3=1-\frac{2}{3}\hat{\Omega}_\mathrm{stiff}$
yields
\begin{equation}
\textfrac{1}{3}(1-2\hat{\Sigma}) \leq p_\alpha \leq
\textfrac{1}{3}(1-\hat{\Sigma}) \leq p_\beta \leq
\textfrac{1}{3}(1+\hat{\Sigma})  \leq p_\gamma \leq
\textfrac{1}{3}(1+2\hat{\Sigma}),
\end{equation}
where $\hat{\Sigma}=\sqrt{1-\hat{\Omega}_{\text{stiff}}}$, and
where $(\alpha\beta\gamma) = (123)$, and cycle. Since we showed
in Proposition \ref{case2} that the past attractor in the stiff
fluid case must be confined to ${\cal J}$ we immediately get
from requiring consistency with the stability analysis that it
must be contained in the closure of the stable part of ${\cal
J}^{\Delta}$, i.e $\overline{{\cal J}^{\Delta}}$.

In the soft case (iii), it follows from~\eqref{geomstabeq} that
${\cal K}$ is unstable everywhere toward the past (except at
the excluded points $p_\mathrm{max}=1$) with an unstable
$N_\alpha$-mode when $\hat{\Sigma}_\alpha<-1$, or,
equivalently, when $p_\alpha<0$, since $p_1+p_2+p_3=1$ and
$p_1^2+p_2^2+p_3^3=1$ yields
\begin{equation}
-\textfrac{1}{3} \leq p_\alpha \leq 0 \leq p_\beta \leq
\textfrac{2}{3} \leq p_\gamma \leq 1,
\end{equation}
where $(\alpha\beta\gamma) = (123)$, and cycle.

The past instabilities in the stiff and soft cases are
associated with so-called silent Bianchi type II curvature
\textit{transitions\/}, to use the nomenclature
of~\cite{heietal09}, i.e., orbits associated with Bianchi type
II. We therefore take a closer look at this silent subset for
the stiff and soft cases in a Fermi-propagated simultaneously
diagonalized $\Sigma_{\alpha\beta}$ and $N_{\alpha\beta}$
frame.

\subsection{The type II subset}\label{typeII}

In the past asymptotic limit $v^\alpha_\mathrm{stiff}=0$ in the
stiff case, and $\Omega_{(i)}=0$ in the stiff and soft cases,
where $i$ refers to a fluid with an asymptotically soft
equation of state. We are thus interested in the subset on
Bianchi type II that is described by a single stiff fluid with
$v^\alpha_\mathrm{stiff}=0$ in the stiff case (ii), and the
vacuum type II subset in the soft case (iii). We choose a
Fermi-propagated shear eigenframe with
$\Sigma_{\alpha\beta}=\text{diag}(\Sigma_1,\Sigma_2,\Sigma_3)$,
and project the dynamics onto
$\Sigma_\alpha$-$\Omega_\mathrm{stiff}$-space, i.e., we
disregard the test fields $v^\alpha_{(i)}$; in addition we set
$N_{\alpha\beta}=0$, except for a single component
$N_{\gamma\gamma}=N_\gamma$, which we determine via the Gauss
constraint, which yields
$N_\gamma^2=12(1-\Sigma^2-\Omega_\mathrm{stiff})$.
Eqs.~\eqref{sigsil} and~\eqref{ompfsil} then yield
\begin{subequations}\label{II}
\begin{align}
\parb_0 (2 - \Sigma_\alpha) &= -(2 - q)(2 - \Sigma_\alpha) ,\\
\parb_0 (2 - \Sigma_\beta) &= -(2 - q)(2 - \Sigma_\beta) ,\\
\parb_0 (4 + \Sigma_\gamma) &= -(2 - q)(4 + \Sigma_\gamma),
\label{IIgamma}\\
\parb_0 \Omega_\mathrm{stiff} &= -2(2 - q)\Omega_\mathrm{stiff} ,
\end{align}
\end{subequations}
where $(\alpha\beta\gamma) =(123)$, and cycle, and where
$q=2(\Sigma^2 + \Omega_\mathrm{stiff})$, where
$\Omega_\mathrm{stiff}=0$ in the vacuum case.

It follows that the solutions to~\eqref{II} are trajectories
that are straight lines when projected onto
$\Sigma_\alpha$-space. Since $-2 \leq \Sigma_\alpha \leq 2$,
$\alpha=1,2,3$, and $q < 2$ on the type II subset,
equations~\eqref{II} show that $\Sigma_\gamma$
($\Sigma_\alpha,\Sigma_\beta$) is monotonically increasing
(decreasing) toward the past and approaches a limit value on
the type I boundary where $q=2$. Together with the previous
stability analysis, this implies that the solutions originate
from Jacobi/Kasner fixed points with $\Sigma_\gamma< -1$ and
end at Jacobi/Kasner fixed points with $\Sigma_\gamma> -1$,
when the direction of time is taken to be toward the
singularity. Hence the global future attractor of~\eqref{II} is
given by the fix points on ${\cal K}\cup {\cal J}$ (${\cal K}$
in the vacuum case) for which $\Sigma_\gamma\leq -1$, while the
past attractor of~\eqref{II} is given by the fixed points on
${\cal K}\cup{\cal J}$ (${\cal K}$ in the vacuum case) with
$\Sigma_\gamma\geq -1$. Using the nomenclature
of~\cite{heiugg09b,heietal09}, the solution trajectories are
denoted as \textit{single curvature transitions}, and they
reflect and describe the outcome of the past instabilities
associated with $N_{\alpha\beta}$ described in the previous
subsection.

\subsection{Past attractors on the silent boundary}\label{pastatt}

Combining the previous linear and non-linear stability analysis
with the results in subsection~\ref{subsec:typeI}, notably the
past asymptotic consequences for $v^\alpha_{(i)}$ that were
obtained via Eq.~\eqref{v2sub}, now allows us to make some firm
statements about the past attractors on the silent boundary in
cases (i) and (ii), and we also make some predictions about the
past attractor in case (iii). Throughout we use an asymptotic
Fermi frame.

\vspace*{2mm} \noindent {\bf The ultra-stiff fluid case (i)}:

\begin{proposition} The ultra-stiff fluid case (i): If the past
attractor ${\cal A}^-_\mathrm{ultra-stiff}$ is contained on the
silent boundary, then it is given by
\begin{displaymath}
{\cal A}^-_\mathrm{ultra-stiff} = {\cal F}^-,
\end{displaymath}
where ${\cal F}^-$ is characterized by
\begin{displaymath}
(\Sigma_{\alpha\beta}, N_{\alpha\beta},A_\alpha)=(0,0,0),\quad
\Omega_\mathrm{tot}=1, \quad Q^\alpha_\mathrm{tot}=0, \quad
\Pi^\alpha_\mathrm{tot}=0 ,
\end{displaymath}
and
\begin{displaymath}
\begin{array}{rl}
\Omega_\mathrm{tot}&=\Omega_\mathrm{ultra-stiff}=1; \quad
v^\alpha_\mathrm{ultra-stiff}=0;\quad \Omega_{(i)}=0;\\
v^\alpha_{(i)} &=0 \,\,\, \text{when}\,\,\, (c_s^2)_{(i)}>\textfrac{1}{3};\quad
v^\alpha_{(i)}=\hat{v}\hat{c}^\alpha_{(i)} \,\,\, \text{when}\,\,\,
(c_s^2)_{(i)}=\textfrac{1}{3};\quad
v^\alpha_{(i)}=\hat{c}^\alpha_{(i)} \,\,\, \text{when}\,\,\,
(c_s^2)_{(i)}<\textfrac{1}{3} .
\end{array}
\end{displaymath}
\end{proposition}

\begin{proof}
This follows directly from Propositions~\ref{case1}
and~\ref{pecdirprop}, and the results that followed from
Eq.~\eqref{v2sub}.
\end{proof}

\vspace*{2mm} \noindent {\bf The stiff fluid case (ii)}:

\begin{proposition}
The stiff fluid case (ii):  If the past attractor ${\cal
A}^-_\mathrm{stiff}$ is contained on the silent boundary, then
it is given by
\begin{equation}
{\cal A}^-_\mathrm{stiff} = (\overline{{\cal J}^\Delta})^- ,
\end{equation}
where $(\overline{{\cal J}^\Delta})^-$ is characterized by
\begin{subequations}
\begin{align}
N_{\alpha\beta} & = 0, \qquad A_\alpha =0,\\
\Sigma_{\alpha\beta} &= \hat{\Sigma}_{\alpha\beta},\qquad
\text{such that}\qquad
\hat{\Sigma}_\alpha \geq -1 \quad (\text{or, equivalently,}\,\, p_\alpha\geq0)
\quad \forall \alpha ,\\
\Omega_\mathrm{tot} & = \hat{\Omega}_\mathrm{stiff}, \qquad Q^\alpha_\mathrm{tot}=0, \qquad
\Pi^\alpha_\mathrm{tot}=0 ,
\end{align}
\end{subequations}
where $(\hat{\Sigma}_1,\hat{\Sigma}_2,\hat{\Sigma}_3)=(3p_1-1,
3p_2-1, 3p_3-1)$ are the (non-ordered) eigenvalues of
$\hat{\Sigma}_{\alpha\beta}$, and (apart from the stiff
fluid(s) for which $\Omega_\mathrm{stiff} =
\hat{\Omega}_\mathrm{stiff} =\Omega_\mathrm{tot}$):
\begin{subequations}\label{pecstiffresult}
\begin{align}
v^\alpha_\mathrm{stiff} &= 0;\qquad \Omega_{(i)}=0;\\
v^\alpha_{(i)} &= 0 \,\,\, \text{when}\,\,\, (c_s^2)_{(i)}>
p_\mathrm{max} = \textfrac{1}{3}(1+ \hat{\Sigma}_\mathrm{max}),\\
v^\alpha_{(i)} &= \hat{v}^\alpha_{(i)} \,\,\, \text{when}\,\,\, (c_s^2)_{(i)}=
p_\mathrm{max} = \textfrac{1}{3}(1+ \hat{\Sigma}_\mathrm{max}),\\
v^\alpha_{(i)} &= \hat{c}^\alpha_{(i)} \,\,\, \text{when}\,\,\, (c_s^2)_{(i)}<
p_\mathrm{max} = \textfrac{1}{3}(1+ \hat{\Sigma}_\mathrm{max}),
\end{align}
\end{subequations}
\end{proposition}

\begin{proof}
Follows directly from Proposition \ref{case2},
Eq.~\eqref{v2sub}, and the linear analysis of
\eqref{geomstabeq}. Note that the velocity directions refer to
the `dominant' shear eigen-directions, as discussed after
Eq.~\eqref{v2sub}.
\end{proof}

\vspace*{2mm} \noindent {\bf The soft fluid case (iii)}: In the
soft fluid case we have no proof, but the previous analysis
suggests that there exists a past attractor ${\cal
A}^-_\mathrm{soft}$ subset, on the vacuum part of the silent
boundary, that describes the asymptotic dynamics of a timeline
in terms of Mixmaster like behavior (see~\cite{waiell97} sec.
6.4), where the asymptotic dynamics is approximated by an
infinite heteroclinic sequence that reside on ${\cal
A}^-_\mathrm{soft}$, which hence breaks asymptotic
self-similarity~\cite{waietal99}; ${\cal A}^-_\mathrm{soft}$ is
given by
\begin{equation}
{\cal A}^-_\mathrm{soft} = {\cal K}\cup {\cal B}_\mathrm{II}^\mathrm{vacuum},
\end{equation}
where $\Sigma_{\alpha\beta}$ on the Kasner subset ${\cal K}$ is
described by
\begin{equation}
\Sigma_{\alpha\beta} = \hat{\Sigma}_{\alpha\beta} ,
\end{equation}
since we use a Fermi frame, and where ${\cal
B}_\mathrm{II}^\mathrm{vacuum}$ is the silent vacuum Bianchi
type II subset.\footnote{This subset consists of the union of
six disjoint Bianchi type II subset representations, each
characterized by the sign of a single non-zero eigenvalue of
$N_{\alpha\beta}$.}

Since the conjectured attractor consists of the vacuum type I
and II subsets, we expect that we (at least generically, recall
footnote 8) can asymptotically diagonalize
$\Sigma_{\alpha\beta}$ and $N_{\alpha\beta}$ (although the
diagonalized shear directions will typically be different for
different timelines). To describe the asymptotic dynamics we
hence perform a constant rotation that diagonalizes an `initial
asymptotic' Kasner point
$\hat{\Sigma}_{\alpha\beta}^\mathrm{i}$, which leads to that
the subsequent dynamics is described by the same heteroclinic
sequence (Kasner states joined by type II curvature
transitions) as in the vacuum Bianchi type VIII and IX cases,
see e.g.~\cite{waiell97}.

In the soft fluid case the description of ${\cal
A}^-_\mathrm{soft}$ also involves the asymptotic test fields
$v^\alpha_{(i)}$. Just as $\Sigma_{\alpha\beta}$ and
$N_{\alpha\beta}$ oscillate perpetually, so do the fields
$v^\alpha_{(i)}$. The effects of a sequence of Kasner
transitions by means of curvature type II transitions is
two-fold: (a) a change of Kasner state, (b) a change of ordered
Kasner shear eigen-directions associated with $p_\alpha\leq
p_\beta \leq p_\gamma$, where $(\alpha\beta\gamma) = (123)$, or
a permutation thereof. The curvature transitions induce a
sequence of `tilt' (peculiar velocity) transitions. As the
dynamics approach the attractor it follows that the asymptotic
vacuum dynamics spend an increasing time near the Kasner fix
points. This leads to that the asymptotic test fields $v_{(i)}$
have increasingly long periods of time to reach their past
asymptotic states on the Kasner subset. This in turn implies
that velocities for increasingly long times are either almost
aligned, anti-aligned, or one or two of the associated speeds
are zero, depending on $(c_s^2)_{(i)}$, $(c_s^2)_{(j)}$ ($i\neq
j$), and $p_\mathrm{max}$, see subsection~\ref{subsec:typeI}.
However, this correlation is temporally broken and changed
whenever there is a curvature transition. But since the
curvature transitions are increasingly dominated in time by the
Kasner states it follows that the probability of finding
$v^\alpha_{(i)}$ and $v^\alpha_{(j)}$ in the previously
described correlated state increases with time.\footnote{
In~\cite{heietal09} it was shown that one can expect that
oscillations yield cumulative trends over time. This also
pertains to the peculiar velocities, and the results
in~\cite{heietal09} suggest that such trends depend on the
stiffness of the equation of state; this is also suggested by
numerical experiments for special models \cite{LimHervik2009}.}

\section{Past stability and instability in the full state
space}\label{stabfull}

We now turn to the discussion of the role of the past
attractors on the silent boundary of the stiff and soft cases
in the full physical state space. We have previously assumed
that the dynamics approach the silent boundary toward the past
and that $E_\alpha{}^i\rightarrow 0$, and we have subsequently
worked out the consequences of these assumptions. To check the
consistency of this it is of interest to compute $E_\alpha{}^i$
`on' the past attractors by inserting the attractor subset
variable values in $F_\alpha{}^\beta$ in the evolutions
equation $\parb_0\, E_\alpha{}^i = F_\alpha{}^\beta
E_\beta{}^i$, thus yielding a lowest order past attractor
perturbation of $E_\alpha{}^i$ in the full state space. Since
the past attractor resides on the type I subset in the stiff
case (ii) and since we expect the type I subset to `dominate'
the Mixmaster dynamics in the soft case (iii) (the ultra-stiff
case (i) has already been discussed previously), we insert the
diagonalized shear values on ${\cal J}, {\cal K}$ in
$F_\alpha{}^\beta$; this yields the following equation for the
individual $E_\alpha{}^i$ components:
\begin{equation}\label{Estab}
(E_\alpha{}^i)^{-1}\,\parb_0\,E_\alpha{}^i|_{{\cal J},{\cal K}}
= 2-\hat{\Sigma}_{\alpha} = 3(1-p_\alpha) ,
\end{equation}
and thus we see that $E_\alpha{}^i$ is stable toward the past
everywhere on ${\cal J}, {\cal K}$, except at the Taub points,
as is to be expected, but which nevertheless yields support for
the assumption $E_\alpha{}^i\rightarrow 0$.

Next we discuss $r_\alpha$. Recall that $r_\alpha =
-E_\alpha{}^i\ptl_i\ln H$, and since $E_\alpha{}^i\rightarrow
0$, then $r_\alpha \rightarrow 0$ if we have chosen a gauge so
that $\ptl_i\ln H$ does not blow up too fast. A way at looking
at the evolution of $r_\alpha$ is to heuristically regard the
evolution equation~\eqref{dlrdot} for $r_\alpha$ asymptotically
as an equation of the form $\parb_0\,r_\alpha=a_\alpha{}^\beta
r_\beta + b_\alpha$, where $a_\alpha{}^\beta$ is
$F_\alpha{}^\beta$ computed on the past attractor, while
$b_\alpha$ is $(\parb_{\alpha} + \Udot_\alpha)(q+1)$ calculated
`on' the past attractor, where $E_\alpha{}^i$ in $\parb_\alpha$
is computed by inserting the attractor values in
$F_\alpha{}^\beta$, which leads to the evolution equation
$\parb_0\,E_\alpha{}^i = a_\alpha{}^\beta\,E_\alpha{}^i$. This
leads to that one can regard $a_\alpha{}^\beta$ and $b_\alpha$
as time dependent coefficients \emph{on} a given timeline,
effectively leading to an ODE for $r_\alpha$ where the solution
for $r_\alpha$ is given by the solution to the homogenous
equation $\parb_0\,r_\alpha = a_\alpha{}^\beta r_\beta$ added
to a particular solution associated with $b_\alpha$. However,
due to that the homogenoeus equation for $r_\alpha$ has the
same character as that for $E_\alpha{}^i$ it follows that the
homogeneous solution tends to zero. Thus we require a gauge
that is such that the particular solution also tends to zero,
where the freedom in the gauge choice is reflected in the term
$b_\alpha$; we expect that we require a gauge such that
$(\parb_{\alpha} + \Udot_\alpha)(q+1)$ tends to zero reasonably
fast. Considering that $q$ is 2 in the stiff case or `almost
always' 2 in the soft case due to `Kasner dominance' this
suggest that this is presumably a rather wide
class.\footnote{In~\cite{andetal05}, where ${\cal
M}_\alpha=W_\alpha=0$, it was noted that $E_\alpha{}^i=0,
\Udot_\alpha=0$ yields an invariant boundary subset, where
$r_\alpha\neq 0$ leads to the same equations as those for
spatially self-similar models. In~\cite{andetal05}, this subset
was referred to as the silent boundary, but since $\parb_0\,
r^2=(q\delta_\alpha{}^\beta -
\Sigma_\alpha{}^\beta)r_\alpha\,r^\beta$ on this subset, which
leads to that $r_\alpha\rightarrow 0$, we have chosen to focus
on the subset with $r_\alpha=0$, which we here has referred to
as the silent boundary. Furthermore, note that $r_\alpha$ is
stable toward the past on the `extended' silent boundary.} For
an example of a gauge with $({\cal M}_\alpha,W_\alpha)=(0,0)$
for which there is numerical support that $r_\alpha\rightarrow
0$ (as well as $\Udot_\alpha\rightarrow 0$),
see~\cite{andetal05}.

We now turn from considering reference congruences in general,
to the issue if there are \emph{fluid congruences} for which
the Hubble normalized vorticity $W^\alpha$ and acceleration
$\Udot^\alpha$ will vanish asymptotically. Choosing the
timelike reference congruence as one of the fluid congruences
implies that for that fluid $v_\alpha=0,\, \rho=\tilde{\rho},\,
p=\tilde{p}$, and $Q_\alpha = \Pi_{\alpha\beta}=0$, while
$P=w\Omega$ (again we drop the index $(i)$). The fluid
equations reduce to (obtained by specializing the total matter
equations~\eqref{dmattereq} in Appendix~\ref{app} to a single
comoving perfect fluid)
\begin{subequations}
\begin{align}
\parb_0\,\Omega &= [2q - 1 - 3w]\,\Omega, \\
0 & = c_s^2\left(\parb_\alpha - 2r_\alpha\right)\Omega +
(1+w)(\Udot_\alpha + r_\alpha)\Omega,
\end{align}
\end{subequations}
or equivalently,
\begin{subequations}\label{rhopeq}
\begin{align}
\parb_0\,\rho &= -3(\rho + p), \\
0 & = \parb_\alpha\,p + (\Udot_\alpha + r_\alpha)(\rho + p).
\end{align}
\end{subequations}
Assuming that the weak energy condition holds strictly for the
fluid component at hand, i.e., $\rho>0$ and $\rho + p>0$, makes
it possible to introduce the particle density $n$ and the
chemical potential $\mu$,
\be \frac{dn}{n} = \frac{d\rho}{\rho + p}, \qquad \mu
=\frac{\rho + p}{n}, \qquad \frac{d\mu}{\mu} = \frac{dp}{\rho +
p}, \ee
which, together with~\eqref{rhopeq}, yields
\begin{subequations}\label{dcomovemeq}
\begin{align}
\parb_0\,n &= -3 n,\\
\quad 0 &= (\parb_\alpha + \Udot_\alpha + r_\alpha)\mu,\label{mu}
\end{align}
\end{subequations}
where a suitable function of $n$ may be useful as a matter
variable in the case $w\neq const$, see~\cite{heietal05}.

By applying $\parb_0$ to~\eqref{mu} and using~\eqref{dlrdot}
and~\eqref{dc0a} we obtain
\begin{equation}\label{udotcomov}
\parb_0\,\Udot_\alpha =
[F_\alpha{}^\beta + (3c_s^2 - 1 -
q)\delta_\alpha{}^\beta]\,\Udot_\beta + \parb_\alpha(3c_s^2 -
q).
\end{equation}
Equations~\eqref{dcomovemeq} and~\eqref{dlrcon} together with
applying~\eqref{dcab} to $\ln \mu$, and using the relation
$d\ln \mu/d\ln n = c_s^2 = dp/d\rho$, yield
\begin{equation}
\textfrac{1}{2}{\bf C}_\alpha{}^\beta\,\Udot_\beta = (3c_s^2 - q
-1)\,W_\alpha,
\end{equation}
which allows equation~\eqref{wevol} to be written on the form
\begin{equation}\label{omegaevolinv}
\parb_0\,W_\alpha  =
(F_\alpha{}^\beta + (3c_s^2 -1)\,\delta_\alpha{}^\beta +
2\Sigma_\alpha{}^\beta)\,W_\beta.
\end{equation}

Following Taub~\cite{taub69,elsugg97}, we let
\be M=\frac{M_0}{\mu}, \ee
where $M_0=M_0(x^0)$, which, via~\eqref{Heq}, \eqref{dMaeq},
and \eqref{deeq} yields that
\be M_i=M_i(x^j)=\hat{M}_i, \ee
which gives that the time dependence of ${\cal M}_\alpha$ is
determined by $E_\alpha{}^i$ since
\begin{equation}\label{Ma}
{\cal M}_\alpha=E_\alpha{}^i\,\hat{M}_i.
\end{equation}
Applying equations~\eqref{dMconst} and~\eqref{deconst} to this
result gives $W_\alpha= \textfrac{1}{2}{\cal M}E_\beta{}^i{\bf
C}_\alpha{}^\beta\,\hat{M}_i$ (a relation that is equivalent to
the non-normalized coordinate frame expression
$\omega_{ij}=M\partial_{[i}\hat{M}_{j]}$). Since Eq.~\eqref{Ma}
implies that if $E_\alpha{}^i\rightarrow 0$ then ${\cal
M}_\alpha\rightarrow 0$ it remains to investigate if $W_\alpha$
and $\Udot_\alpha$ tends to zero toward the past.

In the ultra-stiff case (i) it follows straight forwardly that
$W_\alpha\rightarrow 0$ for the ultra-stiff fluid in the
neighborhood of ${\cal F}$. If in addition
$\parb_\alpha(3(c_s^2)_\mathrm{ultra-stiff}-q)\rightarrow 0$
sufficiently fast, which can be shown to be a consistent
condition by means of an analysis similar to that of other
isotropic singularities undertaken in~\cite{limetal04} (see
also~\cite{collim05}), then also $\Udot_\alpha\rightarrow 0$;
this is to be expected since $v_\mathrm{ultra-stiff}=0$
asymptotically when measured some congruence that is assumed to
satisfy the asymptotic surface formation condition
\eqref{condasila} (we also expect that soft fluids with
$c_s^2>\frac{1}{3}$ in the ultra-stiff case satisfy the
asymptotic surface formation condition since they lead to
$v^\alpha\rightarrow 0$, cf. subsection~\ref{subsec:typeI}).

Let us turn to the stiff (ii) and soft cases (iii). In analogy
with subsection~\ref{perturbI}, let us study the stability of
$W_\alpha$ and $\Udot_\alpha$ by making a perturbation of
${\cal J}^\Delta/{\cal K}$ in a shear diagonalized Fermi frame.
Eq.~\eqref{omegaevolinv} then yields
\begin{equation}\label{Wstab}
W_\alpha^{-1}\parb_0\,W_\alpha|_{{\cal J}^\Delta,{\cal K}}  =
1 + 3c_s^2  +
\hat{\Sigma}_\alpha = 3(c_s^2 + p_\alpha) ,
\end{equation}
which requires $c_s^2+p_\alpha>0\,\forall\,\alpha$ in order for
$W_\alpha\rightarrow 0$.

On ${\cal J}^\Delta$, the stable $\hat{\Sigma}_\alpha$
satisfies $\hat{\Sigma}_\alpha>-1,\,\,\forall \,\alpha$
($p_\alpha>0,\,\,\forall \,\alpha$), and on this part
$W_\alpha\rightarrow 0$ when $c_s^2\geq 0$. In the soft case
(iii) $\min(p_1,p_2,p_3)=-\frac{1}{3}$ on ${\cal K}$ and thus
$c_s^2> \frac{1}{3}$ leads to that $W_\alpha\rightarrow 0$
everywhere on ${\cal K}$, but for fluids with
$c_s^2<\frac{1}{3}$, parts of the ${\cal K}$ become past
unstable with respect to the vorticity, and for dust ($w =
c_s^2 = 0$) all of ${\cal K}$ is unstable (except at the
non-transversally-hyperbolic Taub points); hence the vorticity
of dust does not vanish in the approach to the singularity. For
$0<c_s^2<\frac{1}{3}$ it is the cumulative effect over time of
the factor $c_s^2+p_\alpha$ that matters; to determine this
effect would require a study by means of, for example, methods
used in~\cite{heietal09}, which we will refrain from since it
is not enough that the vorticity tends to zero in order for the
asymptotic surface formation condition~\eqref{condasila} to be
fulfilled, it is also required that $\Udot_\alpha\rightarrow
0$.

The analysis of~\eqref{udotcomov} of the past asymptotic
behavior of $\Udot_\alpha$ is complicated by the term
$\parb_\alpha(3c_s^2 - q)$. However, by considering its
asymptotic expression, by inserting the asymptotics for $q$ and
$c_s^2$, and by solving the evolution equation for
$E_\alpha{}^i$ `on' the silent boundary (i.e., by perturbing
the past attractor to lowest order), this term can be regarded
as a time-dependent inhomogeneous term; similarly one can
compute the factor before $\Udot_\alpha$ on the r.h.s., which
yields an equation of the form $\parb_0\,\Udot_\alpha
=a\Udot_\alpha + b_\alpha$, where $a$ and $b_\alpha$ can be
regarded as given time dependent functions \textit{on} a given
timeline. Hence the general solution can be obtained by adding
a particular solution to the general solution of the
homogeneous part, $\parb_0\,\Udot_\alpha =a\Udot_\alpha$. In
order for the fluid to be asymptotically surface
forming~\eqref{condasila} it is required that
$\Udot_\alpha\rightarrow 0$ \textit{generically}, and a
necessary condition for this is that $\Udot_\alpha\rightarrow
0$ according to the homogeneous equation, which, when computed
in a Fermi frame on ${\cal J}^\Delta$/${\cal K}$, yields
\begin{equation}\label{Udotstab}
\Udot_\alpha^{-1}\parb_0\,\Udot_\alpha|_{{\cal J}^\Delta,{\cal K}} =
3c_s^2 - 1 - \hat{\Sigma}_\alpha = 3(c_s^2 - p_\alpha) .
\end{equation}

In the stiff case (ii) $\Udot_\alpha\rightarrow 0$ requires
that $c_s^2>p_\mathrm{max}$ on ${\cal J}^\Delta$, i.e., the
same condition as required for $v^\alpha\rightarrow 0$ (which
of course is to be expected). The condition that
$\Udot_\alpha\rightarrow 0$ holds \textit{everywhere} on ${\cal
J}^\Delta$, and in this sense holds for a generic solution,
requires that $c_s^2 - p_\mathrm{max}>0$ everywhere on ${\cal
J}^\Delta$, which leads to that $c_s^2=1$, i.e., none of the
softer fluids will in this case fulfill
condition~\eqref{condasila}.

In the soft case (iii) there will always be parts of ${\cal K}$
that are past unstable with respect to the fluid accelerations,
the asymptotic behavior then depends on the cumulative effect
of the factor $(c_s^2 - p_\alpha)$, as it does with the
vorticity; for fluids with $c_s^2<\frac{2}{3}$ all of ${\cal
K}$ is unstable in at least one mode, since one of the
$p_\alpha\geq \frac{2}{3}$; hence the acceleration does not
vanish for any fluid with low sound speed. This in turn
probably leads to that the condition on the particular solution
for $r_\alpha$ breaks down and hence $r_\alpha$ does not tend
to zero either. Hence comoving gauges for fluids with
$c_s^2<\frac{2}{3}$ are not gauges that are compatible with the
asymptotically surface forming condition~\eqref{condasila}, and
it may be that this is also the case when $c_s^2\geq
\frac{2}{3}$ (to provide plausible arguments for this would
require an extensive study using, for example, methods
described in~\cite{heietal09}).

We conclude this section with a discussion of the asymptotic
behavior of $\rho_{(i)}$ and $\tilde{\rho}{(i)}$. The evolution
equation for ${\rm ln}\,\rho_{(i)}$ on the silent class A
subset is governed by the sign of the factor $-[3 + v^2_{(i)} +
\Sigma_{\alpha\beta}\,
c^\alpha_{(i)}\,c^\beta_{(i)}\,v^2_{(i)}] = -[2 + (1 -
v^2_{(i)}) + (2\delta_{\alpha\beta} +
\Sigma_{\alpha\beta})c^\alpha_{(i)}\,
c^\beta_{(i)}\,v^2_{(i)}]$, see Eq.~\eqref{energydensevol}.
Because of Corollary~\eqref{sigmabound}, $v^2_{(i)} \leq 1$
leads to that the factor associated with $\rho_{(i)}$ is
strictly negative, and hence it follows that $\lim_{x^0 \to -
\infty}\rho_{(i)}\to \infty,\ \forall\ i$ if a solution
approaches the attractor on the silent boundary.

The evolution equation for ${\rm ln}\,\tilde{\rho}_{(i)}$ on
the class A boundary is governed by the sign of the factor $-[3
- v^2_{(i)} - \Sigma_{\alpha\beta}\,
c^\alpha_{(i)}\,c^\beta_{(i)}\,v^2_{(i)}] = -[3(1-v^2_{(i)}) +
(2\delta_{\alpha\beta} -
\Sigma_{\alpha\beta})c^\alpha_{(i)}\,c^\beta_{(i)}\,v^2_{(i)}]$,
see Eq.~\eqref{energydensevol}. Lemma \ref{Non-Taub condition}
and $v^2_{(i)} \leq 1$ suggests that $\lim_{x^0 \to -
\infty}\tilde{\rho}_{(i)}\to \infty,\ \forall\ i$, but
unfortunately, we do not have a strict inequality in this case,
and since $\tilde{\rho}_{(i)}$ needs to be evaluated in the
interior physical state space it is not certain that the
quantities that have been neglected in the above equation do
not prevent $\tilde{\rho}_{(i)}\rightarrow\infty$; nevertheless
$\tilde{\rho}_{(i)}\rightarrow\infty$ is normally what is
assumed in a BKL context and we do so here as well. Thus e.g.
when $c_s^2$ appears in an equation that is used in an
asymptotic context it refers to the limit when
$\tilde{\rho}_{(i)}\rightarrow\infty$, and either represents an
actual limit or a bound.

\section{Conclusions}\label{concl}

We have studied the past asymptotic dynamics of spacetimes with
an arbitrary number of perfect fluids with non-zero peculiar
velocities, and with general barotropic equations of state,
where it has been assumed that the Hubble-normalized
interactions can be asymptotically neglected. Using dynamical
systems methods on a system of equations obtained from the 1+3
Hubble-normalized conformal orthonormal frame approach, we have
reformulated two well known conjectures by Belinksii,
Khalatnikov and Lifshitz, about properties at the vicinity of a
generic spacelike cosmological singularity, to conditions on
our variables (`the locality conjecture' (1) and `the matter
does not matter' conjecture (2)) and worked out the
consequences of these assumptions. We have shown that from the
assumption of `the locality conjecture' (1) alone follows:

\begin{itemize}
\item In the case where there exists at least one fluid
    with an equation of state that is ultra-stiff
    asymptotically to the past (i.e., the speed of sound
    $c_s$ satisfies the inequality $c_s^2 > 1$), then the
    Hubble-normalized shear and spatial curvature will
    vanish asymptotically along with the Hubble-normalized
    energy densities $\Omega_{(i)}$ of all the fluids but
    the one with the asymptotically stiffest equation of
    state, which will have a Hubble-normalized density
    parameter $\Omega_\mathrm{ultra-stiff}$ of unity and a
    vanishing peculiar velocity. The peculiar velocities
    $v^\alpha_{(i)}$ for fluids with $\Omega_{(i)}=0$
    asymptotically vanishes in this case when
    $(c_s^2)_{(i)} < 1/3$ while $v^2_{(i)} = 1$ if
    $(c_s^2)_{(i)} > 1/3$.

\item In the case where no fluid is asymptotically
    ultra-stiff, but at least one fluid is asymptotically
    stiff ($w = c_s^2 = 1$), the past asymptotic temporal
    behavior is given by a Jacobs solution. In this case
    $v^\alpha_\mathrm{stiff}=0$ and $\Omega_{(i)}=0$
    asymptotically, $\forall i \neq \mathrm{stiff}$ such
    that $(c_s^2)_{(i)}<1$. The peculiar velocities
    $v^\alpha_{(i)} \rightarrow 0$ tend to zero when
    $(c_s^2)_{(i)} > p_{\mathrm{max}}$, where
    $p_{\mathrm{max}}$ is the maximal shape parameter,
    while $v^2_{(i)} \rightarrow 1$ when $c_s^2 <
    p_{\mathrm{max}}$.

\item In the case when there are no stiff or ultra-stiff
    fluids, but an arbitrary number of fluids with soft
    equations of state ($c_s^2 < 1$), the past asymptotic
    state resides on the union of the Bianchi type I, II,
    VI$_0$ or VII$_0$ subsets on the silent boundary.
    Furthermore, in this case we have made the additional
    assumption that `the matter does not matter' conjecture
    (2) holds, and have provided some arguments that the
    past attractor is Mixmaster like, with oscillations
    between Bianchi type I and II vacuum solutions on the
    silent boundary. The peculiar velocities of the
    fluids---which become test field---become forever
    oscillating, both in direction and amplitude as the
    oscillations change the shear via the so-called Kasner
    map (see e.g.~\cite{heiugg09b}). Interestingly the
    extreme properties of the asymptotic spacetime geometry
    induces `correlation effects' among the different
    peculiar velocities.
\end{itemize}

By studying the vorticity and acceleration in a fluid comoving
frame we come to the conclusion that the fluid comoving gauges
for stiff and ultra-stiff fluids obey the `locality conjecture'
(1) and are therefore acceptable gauges. However, this is not
the case for models with only fluids with asymptotically soft
equations of state. In such models fluid comoving gauges for
fluids with $c_s^2 < 2/3$, which notably include dust and
radiation equations of state, are not compatible with the gauge
requirement associated with `the locality conjecture' (1), and
fluid comoving gauges for fluids with $2/3 \leq c_s^2 < 1$ may
be inadmissible as well, although the latter is an open issue.

In all, our results for spacetimes with multiple fluids agree
with previous studies of special models with soft
(e.g.~\cite{uggetal03,hewetal01,her04}), stiff
(e.g.~\cite{andren01}), and ultra-stiff single fluid
models~\cite{collim05}. However, we here studied asymptotic
dynamics in a general infinite-dimensional dynamical systems
setting, where we pursued the consequences of BKL-like
assumptions, and this led e.g. to the conclusion that comoving
fluid gauges are, for the most physically interesting fluid
cases, incompatible with BKL-like behavior. It follows that a
matter element will always move w.r.t. to a frame that obeys
the gauge requirements of `the locality conjecture' (1),
furthermore, even in the fluid comoving gauge a matter element
will accelerate and pick up momentum w.r.t. the rest frame of
the fluid. In this sense matter momentum will matter toward the
singularity (which is not the case for fluids with stiff or
ultra-stiff equations of state), even though `matter does not
matter' for the asymptotic spacetime geometry, answering a
speculation posed in~\cite{uggetal03}. It was also beneficial
to consider multiple fluids, since this made it possible to
investigate the relative evolution of the fluid themselves with
some interesting results, like the dominance of the stiffest
fluid, Eq.~\eqref{rhocomp}, and the peculiar velocity shear
alignment in type I (Corollary~\ref{velocityshearalignment}),
which led to suggestive results about asymptotic correlations
between different peculiar velocities for asymptotic
oscillatory behavior, as discussed in subsection~\ref{pastatt}.

Our analysis has rested on the assumption that the silent
boundary is approached, and in the soft fluid case on the
further assumption of asymptotic vacuum dominance, and even
though we found support for our assumptions it would be
desirable for further study to establish firm results on all
points. We therefore list three open problems.

\begin{itemize}
\item There seems to exist an intricate connection between
    asymptotically approaching the Taub subset, described
    in Appendix~\ref{invbound}, and the violation of
    lemma~\ref{Non-Taub condition} and the
    condition~\eqref{positivity condition}. Furthermore,
    in~\cite{heietal09} it was shown that one statistically
    with increasingly probability find the state of a
    solution in a small neighborhood of the Taub subset in
    the approach to an oscillating singularity. Moreover,
    there seems to be a connection between weak null
    singularities and the Taub subset~\cite{limetal06}.
    Hence there is a need for a detailed separate study of
    the Taub subset, but such an analysis is unfortunately
    likely to pose a major challenge.

\item Another less formidable future possibility is to
    apply the methods in~\cite{heietal09} to study
    cumulative trends for peculiar velocities, discussed in
    subsection~\ref{pastatt}.

\item A third possible investigation would be to
    investigate what BKL-like conjectures, analogous to the
    presently formulated ones, would imply for other
    sources. When does `matter does not matter' toward the
    initial singularity´ hold in this more general context?
\end{itemize}

\subsection*{Acknowledgments}
It is a pleasure to thank John Wainwright, Henk van Elst, and
Woei Chet Lim, who all have been crucial in producing material
that has served as the foundation for the present paper. CU is
supported by the Swedish Research Council.

\begin{appendix}

\section{The 1+3 conformally Hubble-normalized dynamical systems approach}\label{app}

To establish conventions and notation, we in this Appendix
briefly introduce the conformal 1+3 Hubble-normalized dynamical
systems approach. This constitutes a specialization of the
results in~\cite{rohugg05}, in combination with that we derive
the general perfect fluid equations; for further details and
motivation we refer to~\cite{rohugg05}.

In the conformal Hubble-normalized orthonormal frame approach,
cf.~\cite{rohugg05,ugg08}, we introduce {\em a conformal
`Hubble-normalized' orthonormal frame\/} of ${\bf g}$ (or,
equivalently, an orthonormal frame of ${\bf G}$) according to
${\bf g} = H^{-2}\,{\bf G} =
H^{-2}\,\eta_{ab}\,\bOm^a\,\bOm^b$, where the one-forms
$\bOm^a$ are related to the conformal orthonormal vector fields
$\parb_a$ via $\langle\,\bOm^a,\,\parb_b\,\rangle =
\delta^{a}{}_{b}$.
We align $\parb_0$ with a timelike reference congruence, which
leads to that $\parb_0$ and $\parb_\alpha$ are given by:
\begin{equation}
\parb_0 = H^{-1}\,\vece_0 = {\cal M}^{-1}\ptl_{x^0}, \qquad \parb_\alpha =
H^{-1}\,\vece_\alpha = {\cal M}_\alpha\,{\cal M} \parb_0 + E_\alpha{}^i \ptl_i,
\end{equation}
where ${\cal M}$ and ${\cal M}_\alpha$ are the conformally
Hubble-normalized threading lapse function and shift vector,
respectively; $x^0$ denotes the time coordinate along the
timelike reference congruence, while $\ptl_i=\partial_{x^i}$,
where $x^i$ are spatial coordinates ($i=1,2,3$). Throughout we
express the derivatives in all equations by means of the
derivative operators $\parb_0$ and $\parb_\alpha$. Partial
derivatives are thus `weighted' with conformally normalized
frame variables, and this is one of the main advantages of the
present formalism.

The deceleration parameter $q$ and $r_\alpha$ are objects that
are kinematically defined by
\begin{equation}\label{Heq}
\parb_{0}H = -\,(q+1)\,H,
\qquad\qquad \parb_{\alpha}H = -\,r_{\alpha}\,H.
\end{equation}
For dimensional reasons, the above equations for the
dimensional Hubble variable $H$, associated with the timelike
reference congruence in the physical spacetime associated with
${\bf g}$, must decouple from all equations that only involve
dimensionless variables and operators.

It is useful to write the dimensionless commutator equations on
the following operator form:
\begin{subequations}\label{dcommutator}
\begin{align}
\label{dc0a} 0 & = (\parb_\alpha + \Udot_\alpha)\parb_0 -
(\delta_{\alpha}{}^{\beta}\,\parb_0
- F_{\alpha}{}^{\beta})\,\parb_{\beta},\\
\label{dcab} 0 & = 2W_{\alpha}\,\parb_0 -
\vec{C}_{\alpha}{}^{\beta}\,\parb_{\beta},
\end{align}
\end{subequations}
where
\begin{subequations}\label{FCdef}
\begin{align}
F_{\alpha}{}^{\beta} &=  -[{\cal H}\,\delta_{\alpha}{}^{\beta} +
\Sigma_{\alpha}{}^{\beta} + \epsilon_{\alpha}{}^{\beta}{}_{\gamma}\,(W^{\gamma}+R^{\gamma})]
= q\,\delta_{\alpha}{}^{\beta} -
\Sigma_{\alpha}{}^{\beta} - \epsilon_{\alpha}{}^{\beta}{}_{\gamma}\,(W^{\gamma}+R^{\gamma}),\\
\vec{C}_{\alpha}{}^{\beta} &= \epsilon_{\alpha}{}^{\gamma\beta}\,
(\parb_\gamma - A_\gamma) - N_{\alpha}{}^{\beta},
\end{align}
\end{subequations}
where we have used $\parb_{0}H = -\,(q+1)\,H$ to obtain the
relationship $q=-{\cal H} = - \frac{1}{3}\Theta$, which relates
the deceleration parameter $q$ to the (Hubble-) conformal
Hubble scalar ${\cal H}$ and expansion $\Theta$; for the
physical interpretation of the other quantities, see
Section~\ref{intro}.

We will be concerned with general relativity and hence we
impose Einstein's field equations:
\be G_{ab}=T_{ab}, \ee
where we have chosen $c=1=8\pi G$ as units, where $c$ is the
speed of light in vacuum and $G$ is Newton's gravitational
constant.

We make a 1+3 split of the stress-energy tensor $T_{ab}$ w.r.t.
the tangential 4-velocity $u^a$ of the reference congruence in
the physical spacetime ${\bf g}$ according to:
\begin{subequations}\label{emomcons}
\begin{align}
T_{ab} &= \rho\,u_a\,u_b + 2q_{(a}\,u_{b)} + p\,h_{ab} +
\pi_{ab},\\
h_{ab} &= u_a\,u_b + g_{ab} \quad \Rightarrow \quad h_{ab}\,u^b
=0; \qquad q_a\,u^a=0, \qquad \pi_{ab}u^a = 0,\qquad
\pi^a{}_a=0,
\end{align}
\end{subequations}
and hence the total stress-energy is encoded in the objects
$(\rho, p, q_\alpha, \pi_{\alpha\beta})$, where
$\pi^\alpha{}_\alpha=0$.

The conformal transformation naturally yields new dimensionless
matter variables by scaling $\rho, p, q_\alpha,
\pi_{\alpha\beta}$ with $H^{-2}$, however, to conform with the
standard definition $\Omega = \rho/(3H^2)$, we instead scale
the matter variables as follows:
\begin{equation}\label{dlcurv}
\{\Om, \,P, \,Q^{\alpha}, \,\Pi_{\alpha\beta}\} = \{\rho, \,p,
\,q^{\alpha}, \,\pi_{\alpha\beta}\}/(3H^{2}).
\end{equation}

The field equation for the dimensionless frame and commutator
variables (obtained from the commutator equations, the Jacobi
identities, and the Einstein equations) are conveniently
grouped into evolution equations, and constraint equations:

\vspace*{2mm} \noindent {\em Evolution equations}:
\begin{subequations}\label{devoleq}
\begin{align}
\parb_0\, {\cal M}_\alpha &= F_\alpha{}^\beta\,{\cal M}_\beta + (\parb_\alpha +
\Udot_\alpha){\cal M}^{-1},\label{dMaeq}\\
\parb_0\,W_\alpha &= (F_\alpha{}^\beta +
q\delta_\alpha{}^\beta + 2\Sigma_\alpha{}^\beta)\,W_\beta +
\textfrac{1}{2}{\bf
C}_\alpha{}^\beta\,\Udot_\beta,\label{wevol} \\
\parb_0\, E_\alpha{}^i &= F_\alpha{}^\beta\, E_\beta{}^i,\label{deeq}\\
\parb_0\, \Sigma_{\alpha\beta} &= -(2-q)\Sigma_{\alpha\beta} +
2\epsilon^{\gamma\delta}{}_{\la \alpha}\,\Sigma_{\beta\ra
\delta}\,R_\gamma - \, {}^3\!{\cal S}_{\alpha\beta} +
3\Pi_{\alpha\beta}
- 2W_{\la\alpha}\,R_{\beta\ra} \nonumber\\
&\quad\, + (\parb_{\la\alpha}  + \Udot_{\la \alpha} + A_{\la
\alpha})\,\Udot_{\beta\ra} + 2(\parb_{\la\alpha}  - r_{\la \alpha}
+ A_{\la\alpha})\,r_{\beta\ra} - \epsilon^{\gamma\delta}{}_{\la
\alpha}\,N_{\beta\ra\gamma}\,(\Udot_\delta +2r_\delta),\label{dsig}\\
\parb_0\,A_\alpha &= F_\alpha{}^\beta\, A_\beta +
\textfrac{1}{2}(\parb_\beta +
\Udot_\beta)(3q\delta_\alpha{}^\beta - F_\alpha{}^\beta),\\
\parb_0\, N^{\alpha\beta} &= (3q\delta_\gamma{}^{(\alpha} - 2F_\gamma{}^{(\alpha}) N^{\beta )\gamma} +
\epsilon^{\gamma\delta (\alpha} (\parb_\gamma + \Udot_\gamma)
F_\delta{}^{\beta )},
\end{align}
\end{subequations}
where
\begin{equation}\label{Fdef}
F_{\alpha}{}^{\beta} =  q\,\delta_{\alpha}{}^{\beta} -
\Sigma_{\alpha}{}^{\beta} - \epsilon_{\alpha}{}^{\beta}{}
_{\gamma}\,(W^{\gamma}+R^{\gamma}).
\end{equation}

\vspace*{2mm} \noindent {\em Constraint equations}:
\begin{subequations}\label{dconstreq}
\begin{align}
0 &= {\bf C}_\alpha{}^\beta {\cal M}_\beta - 2{\cal
M}^{-1}W_\alpha,\label{dMconst}\\
0 &= (\parb_\alpha - \Udot_\alpha - 2A_\alpha)\,W^\alpha.\\
0 &= {\bf C}_\alpha{}^\beta\, E_\beta{}^i,\label{deconst}\\
0 &= 1 - \Sigma^2 - \Omega_k - \Omega + \textfrac{1}{3}W^2 -
\textfrac{2}{3}W_\alpha R^\alpha -
\textfrac{1}{3}(2\parb_\alpha - 4A_\alpha + r_\alpha)\,r^\alpha,\label{dGauss}\\
0 &= (3\delta_\alpha{}^\gamma\,A_\beta +
\epsilon_\alpha{}^{\delta\gamma}
\,N_{\delta\beta})\,\Sigma^\beta{}_\gamma - 3Q_\alpha
-(\parb_\beta + 2r_\beta)\,\Sigma_\alpha{}^\beta - [{\bf
C}_\alpha{}^\beta +
2\epsilon_\alpha{}^{\gamma\beta}(\Udot_\gamma + r_\gamma)]\, W_\beta -2r_\alpha,\label{dCodazzi}\\
0 &= A_\beta\, N^\beta{}_\alpha -
\textfrac{1}{2}\parb_\beta\,(\epsilon_\alpha{}^{\beta\gamma} A_\gamma
+ N_\alpha{}^\beta) - (F_\alpha{}^\beta - 2q\delta_\alpha{}^\beta
+
2\Sigma_\alpha{}^\beta)\,W_\beta, \label{dC4}
\end{align}
\end{subequations}
where $\Sigma^2=\frac{1}{6}\Sigma_{\alpha\beta}
\Sigma^{\alpha\beta}$, $W^2=W_\alpha W^\alpha$, and where

\begin{subequations}\label{threecurv}
\begin{align}
\vec{C}_{\alpha}{}^{\beta} &= \epsilon_{\alpha}{}^{\gamma\beta}\,
(\parb_\gamma - A_\gamma) - N_{\alpha}{}^{\beta}, \\
q &= 2\Sigma^{2} + \textfrac{1}{2}(\Om+3P) - \textfrac{2}{3}\,W^{2} -
\textfrac{1}{3} [\parb_\alpha + \Udot_\alpha -2(A_\alpha -
r_\alpha)]\,(\Udot^\alpha + r^\alpha),\label{q}\\
{}^3\!{\cal S}_{\alpha\beta} &= B_{\la \alpha\beta \ra} +
2\epsilon^{\gamma\delta}{}_{\la
\alpha}\,N_{\beta\ra\delta}\,A_\gamma + \parb_\gamma
(\delta^\gamma{}_{\la \alpha}\, A_{\beta\ra} +
\epsilon^\gamma{}_{\la\alpha}{}^\delta\,N_{\beta\ra\delta}),\\
\Omega_k &= -\textfrac{1}{6}\,{}^3\!{\cal R};\qquad
{}^3\!{\cal R} = -\textfrac{1}{2}B^\alpha{}_\alpha - 6A^2 +
4\parb_\alpha\,A^\alpha,\\
B_{\alpha\beta} &= 2 N_{\alpha\gamma}\,N^\gamma{}_\beta -
N^\gamma{}_\gamma\,N_{\alpha\beta},
\end{align}
\end{subequations}
where $A^2=A_\alpha A^\alpha$.\footnote{In~\cite{rohugg05}
there are sign errors in front of the terms
$\epsilon_\alpha{}^{\beta\gamma}R_\beta W_\gamma$,
$\epsilon^{\gamma\delta}{}_{\la \alpha}N_{\beta \ra \gamma}
\Udot_\delta$, and $\epsilon_\alpha{}^{\beta\gamma}R_\beta
W_\gamma$, in the equations that corresponds to~\eqref{dC4},
\eqref{dsig}, and~\eqref{wevol}, respectively.} The expression
for $q$ in~\eqref{q} was obtained from the Raychadhuri
equation, which gives $q$ its dynamical content; ${}^3\!{\cal
S}_{\alpha\beta}, {}^3\!{\cal R}$ can be interpreted as the
trace-free and scalar parts, respectively, of the
Hubble-normalized three-curvature, if the reference congruence
is hypersurface forming ($W_\alpha=0$). The notation
$\la...\ra$ stands for the trace-free part of a symmetric
spatial tensor, i.e. $A_{\la \alpha\beta \ra}=A_{\alpha\beta} -
\frac{1}{3}\delta_{\alpha\beta}\,A^\gamma{}_\gamma$.

The equations~\eqref{dsig}, \eqref{dGauss}, \eqref{dCodazzi},
\eqref{q}, were all obtained from Einstein's field equations,
and are thus dynamical in nature, furthermore, note that it is
the total stress-energy content $\{\Om, \,P, \,Q^{\alpha},
\,\Pi_{\alpha\beta}\}$ that enters into these equations; all
remaining equations were obtained from the commutator equations
and the Jacobi identities, and are thus kinematical. If we want
to stress that a quantity refers to the total stress-energy
content below we will provide it with the subscript
$\mathrm{tot}$, e.g., $\Omega_\mathrm{tot}$.

These equations need to be supplemented with matter equations
that depend on the chosen matter content, however, local
conservation of the total energy-momentum yields
$\nabla_{b}T^{ab} = 0$ for the total $T^{ab}$, which for the
1+3 splitted matter variables yields:\footnote{Note that this
is also a reasonable demand in the context of other metric
theories than general relativity, i.e.,~\eqref{dmattereq} has a
broader area of application than the present general
relativistic one.} \enl

\noindent {\em Total matter equations\/}:
\begin{subequations}\label{dmattereq}
\begin{align}
\label{dlomdot}
\parb_0\,\Om &=  (2q-1)\,\Om - 3P + 2A_{\alpha}\,Q^{\alpha} - \Sigma_{\alpha\beta}\Pi^{\alpha\beta}
-[\parb_\alpha + 2(\Udot_\alpha + r_\alpha)]\,Q^\alpha,\\
\label{dlqmalpha}
\parb_0\, Q_{\alpha} &=  (F_{\alpha}{}^{\beta} - (2-q)\,\delta_{\alpha}{}^{\beta})\,Q_{\beta}
+ (3\delta_\alpha{}^\gamma\,A_\beta +
\epsilon_\alpha{}^{\delta\gamma}
\,N_{\delta\beta})\,\Pi^\beta{}_\gamma \nonumber\\
& \quad+ 2\epsilon_{\alpha}{}^{\beta\gamma}\,W_{\gamma}\,Q_{\beta} -
(\parb_\beta + \Udot_\beta + 2r_\beta)\,(P\delta_\alpha{}^\beta +
\Pi_\alpha{}^\beta) - \Udot_\alpha\,\Omega - r_\alpha (\Omega -
3P).
\end{align}
\end{subequations}

It is sometimes useful to apply the commutator
equations~(\ref{dcommutator}) to $\log (H)$ and consider the
following resulting auxiliary equations for $r_\alpha$ (also
possibly extending the above state space to include
$r_\alpha$):
\begin{subequations}\label{req}
\bea \label{dlrdot}
\parb_{0}r_{\alpha}
& = & F_{\alpha}{}^{\beta}\,r_{\beta}
+ (\parb_{\alpha} + \Udot_\alpha)(q+1), \\
\label{dlrcon} 0 & = & \vec{C}_{\alpha}{}^{\beta}\,r_{\beta} -
2(q+1)\,W_{\alpha}. \eea
\end{subequations}

In this paper we assume that the matter consists of several
perfect fluids with general barotropic equations of state. The
stress-energy component of the $i$:th fluid satisfies, $\bna_a
T^{ab}_{(i)}=I_{(i)}^b$, where $I_{(i)}^b$ represents the
non-gravitational interaction term of the $i$:th fluid with the
other fluids; since $\bna_a T^{ab}_\mathrm{tot}=0$ it follows
that $\sum_i I_{(i)}^a=0$. We are here going to assume that the
Hubble-normalized interaction terms asymptotically tend to zero
toward the singularity, and that the fluids, in this sense, are
asymptotically non-interacting (this can still be the case even
if the interaction energies tend to infinity). Using the
Hubble-normalized version of the relation $\bna_a
T^{ab}_{(i)}=0$, since we assume that the Hubble-normalized
interaction terms are asymptotically zero, leads to the
following equations for $\Omega$ and $v_\alpha$ (to obtain less
cumbersome expressions we drop the index $(i)$):
\begin{subequations}\label{perf}
\begin{align}
\parb_0\,\Omega &= (2q - 1 - 3w)\,\Omega + [(3w-1)\,v_\alpha -
\Sigma_{\alpha\beta}\,v^\beta +
2(A_\alpha - \Udot_\alpha - r_\alpha)  - \parb_\alpha]\,Q^\alpha,\label{omperf}\\
\parb_0 v_{\alpha} &= \bar{G}_-^{-1}\,\left[ (1-v^2)(3c_s^2 - 1 -
c_s^2\,A^\beta\,v_\beta) + (1-c_s^2)(A^\beta +
\Sigma_\gamma{}^\beta\,v^\gamma)\,v_\beta \right] v_\alpha
\nonumber \\
& \quad - [\Sigma_\alpha{}^\beta +
\epsilon_\alpha{}^{\beta\gamma}\,(R_\gamma +
N_\gamma{}^\delta\,v_\delta)]\,v_\beta - A_\alpha\,v^2 +
\epsilon_\alpha{}^{\beta\gamma} W_\gamma\,v_\beta \nonumber \\
& \quad - (\delta_\alpha{}^\beta - v_\alpha\,v^\beta) \Udot_\beta -
(1+w)^{-1}(1-v^2)[(1-w)\delta_\alpha{}^\beta -
4w\,c_s^2\,\bar{G}_-^{-1}\,v_\alpha v^\beta] r_\beta \nonumber \\
& \quad - \left(\frac{v}{Q}\right)\left[(\delta_\alpha{}^\beta +
2c_s^2\,\bar{G}_-^{-1}\,v_\alpha v^\beta)\parb_\gamma
(P\delta_\beta{}^\gamma + \Pi_\beta{}^\gamma) -
(1+c_s^2)\bar{G}_-^{-1}\,v_\alpha\,\parb_\beta Q^\beta \right],
\end{align}
\end{subequations}
where $\bar{G}_- = 1- c_s^2\,v^2$, $c_s^2 =
d\tilde{p}/d\tilde{\rho}$. In the above expressions all `matter
objects' refer to the $i$:th fluid component, except in $q$
in~\eqref{omperf}, since $q$ obtains its dynamical content from
the total source. A cosmological constant $\Lambda$ can
formally be regarded as a perfect fluid contribution with
$w=-1$, which leads to the following Hubble-normalized
stress-energy contribution: $\Omega_\Lambda =
\Lambda/(3H^2)=-P_\Lambda$, while
$Q^\alpha_\Lambda=0=\Pi^{\alpha\beta}_\Lambda$. Due to its
definition and equation~\eqref{Heq}, $\Omega_\Lambda$ satisfies
$\parb_0 \Omega_\Lambda = 2(1+q)\Omega_\Lambda$,\ $\parb_\alpha
\Omega_\Lambda = 2r_\alpha \Omega_\Lambda$.

\section{Invariant boundary subsets}\label{invbound}

The physical interior of the state space is characterized by
$\mathrm{det}(E_\alpha{}^i)\neq 0$, and, in the case of the
interior of the perfect fluid state space,
$\Omega_\mathrm{tot}\neq 0$. However, the asymptotes of most
interior solutions reside on the boundaries of the interior
state space, and hence it becomes necessary to study the
dynamics on these boundaries as well. Some of these boundaries
play a particularly important role. Notably we have the
\textit{vacuum subset\/}  $\Omega_\mathrm{tot} = 0$ (and hence
$\Omega_{(i)}=0 \,\,\forall\,\, i$), and what we call the
partially silent and the silent boundary
subsets~\cite{limetal06}. The existence of these latter subsets
is intimately connected with the homogeneity of~\eqref{deeq},
which leads to the existence of boundary subsets of the
interior subset ($\mathrm{det}(E_\alpha{}^i)\neq 0$) such that
the rank of the matrix $E_\alpha{}^i$ is two, one, or zero.

Let us begin with the rank zero case. Our later discussion
suggests that only a part of the subset $E_\alpha{}^i=0$ is of
generic importance, namely the invariant boundary subset
\be (E_{\alpha}{}^{i}, {\cal M}_\alpha, W_\alpha, \Udot_\alpha,
r_\alpha) = 0, \ee
see~\eqref{devoleq} --~\eqref{dmattereq}, which we denote as
the {\em silent boundary\/}, where ${\cal M}_\alpha =
0,\,E_{\alpha}{}^{i}=0$ yields $\parb_\alpha=0$. On this
subset, described by a state vector ${\bf S}$ (see
Eq.~\eqref{statevec} for the case of several perfect fluids),
there exists a coupled set of ordinary differential equations
and algebraic constraints that are identical to those of
spatially homogeneous models. This can be seen as follows. In
the spatially homogeneous case a spatially homogeneous
foliation with orthogonal timelines (${\cal
M}_\alpha=W_\alpha=0$) leads to $(M, H,{\bf S})= (M(x^0),
H(x^0), {\bf S}(x^0))$, and hence $\Udot_\alpha=r_\alpha=0$ and
$\parb_{\alpha} {\bf S}=E_{\alpha}{}^i\ptl_i\,{\bf S}=0$, and
as a consequence the equations for $E_\alpha{}^i$
($\mathrm{det}(E_\alpha{}^i)\neq 0$) decouple from the rest of
the variables in ${\bf S}$, and thus one often only considers
the equations for the `essential' variables of the state vector
${\bf S}$, cf.~\cite{heiugg09c}. Although the equations for
${\bf S}$ coincide for the spatially homogeneous case and the
silent boundary, there is a fundamental difference; in the
spatially homogeneous case the constants of integration are
really constants, but on the silent boundary the integration
coefficients are spatial functions, since the state space in
this case corresponds to an infinite set of identical
copies---one for each spatial point.

A similar phenomenon happens when the rank of the matrix
$E_\alpha{}^i$ is one or two, which leads to boundary subsets
on which the dynamics is identical to that of models with
spatial symmetry orbits of dimensions two or one, respectively.
We refer to these subsets as \textit{partially silent
boundaries\/}; in these cases there are two or one spatial
coordinates, respectively, that act as an index set, in analogy
with what happens for the state vector in the silent boundary
case.

Yet another, overlapping, boundary is of interest---the
\textit{Minkowski subset\/}. In the present formulation this
subset corresponds to the Minkowski solution/spacetime in
foliations for which $H>0$. Hence it is characterized by that
the Hubble-normalized curvature is zero, i.e., both
$\Omega_\mathrm{tot}$ and the Hubble-normalized Weyl tensor are
zero.

There are many subsets on the Minkowski boundary, but one seems
to be of particular importance, the \textit{Taub subset\/} (so
denoted because it is related to the Taub representation of the
Minkowski spacetime), which we define as a subset that, in
addition to $\Omega_\mathrm{tot}=0$, satisfies ${}^3\!{\cal
R}=0$, i.e., $\Omega_k=0$, and $({\cal M}_\alpha, W_\alpha,
\Udot_\alpha, r_\alpha) = 0$, and hence $\Sigma^2=1$ and $q=2$.
Furthermore, these conditions implies
$\det(\Sigma_{\alpha\beta})=2$, and that it is possible to
introduce a Fermi propagated frame in which
$\Sigma_{\alpha\beta}={\rm diag}(2,-1,-1)$, or cycle.

\end{appendix}

\bibliographystyle{plain}

\end{document}